\documentclass[pdflatex,sn-mathphys-num]{sn-jnl}

\usepackage{graphicx}%
\usepackage{fullpage}
\usepackage{multirow}%
\usepackage{amsmath,amssymb,amsfonts}%
\usepackage{amsthm}%
\usepackage{mathrsfs}%
\usepackage[title]{appendix}%
\usepackage{xcolor}%
\usepackage{textcomp}%
\usepackage{manyfoot}%
\usepackage{booktabs}%
\usepackage{algorithm}%
\usepackage{algorithmicx}%
\usepackage{algpseudocode}%
\usepackage{listings}%
\usepackage{physics}
\usepackage{subcaption}
\usepackage{todonotes}

\theoremstyle{plain}
\newtheorem{theorem}{Theorem}[section]

\newtheorem{lemma}[theorem]{Lemma}

\theoremstyle{definition}
\newtheorem{definition}[theorem]{Definition}

\theoremstyle{remark}

\usepackage{graphicx}
 \usepackage[T1]{fontenc} 

\DeclareMathOperator{\Poly}{poly}

\DeclareMathOperator{\Haar}{Haar}

\DeclareMathOperator{\poly}{poly}

\newcommand{\htwoo}{\ensuremath{\texttt{H}_2\texttt{O}}}

\newcommand{\behtwo}{\ensuremath{\texttt{BEH}_2}}

\begin{document}


\title[Generative replay mitigates sample starvation in quantum architecture search]{\textbf{Gen}erative Replay Mitigates Sample Starvation in \textbf{Q}uantum \textbf{A}rchitecture \textbf{S}earch}


\author*[1,2]{\fnm{Akash} \sur{Kundu}}\email{a.kundu@tudelft.nl}
\equalcont{These authors contributed equally to this work.}

\author[3]{\fnm{Amit Kumar} \sur{Jaiswal}}\email{amit.chr@iitbhu.ac.in}
\equalcont{These authors contributed equally to this work.}

\author[1,2]{\fnm{Sebastian} \sur{Feld}}\email{s.feld@tidelft.nl}
\author[4]{\fnm{Prayag} \sur{Tiwari}}\email{prayag.tiwari@hh.se}

\affil*[1]{\orgdiv{Delft University of Technology}, \orgaddress{\city{Delft}, \country{The Netherlands}}}
\affil*[2]{\orgdiv{Quantum Computing Division, QuTech}, \orgaddress{\city{Delft}, \country{The Netherlands}}}

\affil[3]{\orgname{Indian Institute of Technology (BHU)}, \orgaddress{\city{Varanasi}, \country{India}}}
\affil[4]{\orgname{Halmstad University}, \orgaddress{\city{Halmstad}, \country{Sweden}}}

\abstract{
Reinforcement learning (RL) can automate quantum architecture search, but its scalability is limited when useful circuit trajectories become rare in the rapidly expanding search space. Existing replay mechanisms reuse observed transitions; the proposed learned model produces additional predicted one step transitions from real state–action seeds. Here we introduce GenQAS, a tensor network-guided RL framework that combines a fixed matrix product state warm-start with prioritized generative replay. A learned local transition model generates synthetic circuit transitions on demand and mixes them with real experience during Double Deep Q-Network updates. Under a random exploration analysis, near ground state circuits occupy a rapidly shrinking region of the accessible state space. We investigate whether real data anchored synthetic replay can improve the effective training signal in this regime. Across chemical Hamiltonian benchmarks from 6 to 12 qubits, GenQAS improves fixed-budget success probability and identifies compact circuits at competitive energy error. At 12 qubits, it improves final
success probability by up to $7.0\times$ over passive replay. On a 15-qubit transverse field Ising model, GenQAS increases success probability from $12\%$ to $21\%$. In a noisy 6-qubit BeH$_2$ transfer experiment, generative replay reduces the steps to chemical accuracy by $92.7\%$. These results show that generative replay can mitigate sample starvation in quantum architecture search and support more resource efficient circuit discovery.
}



\maketitle


\section{Introduction}

The performance of near-term and early fault tolerant quantum algorithms are limited by qubit count, imperfect gate execution and hardware noise~\cite{preskill2018quantum, eisert2025mind, katabarwa2024early}. Gate-based quantum processors provide a leading platform for their implementation~\cite{tacchino2020quantum,kjaergaard2020superconducting,wright2019benchmarking,fernandez2026running,wang2024operating,petit2022design,tosato2026crossbar}. However, increasing the number of physical qubits alone does not
guarantee practical utility, as achievable performance also depends on error rates, connectivity, control quality, and the reliably executable circuit depth~\cite{blume2020volumetric,cross2019validating}. This issue is particularly important for variational quantum algorithms~\cite{cerezo2021variational}, in which a parameterized quantum circuit (PQC)
$U(\boldsymbol{\theta},\mathcal{C})$, defined by a discrete architecture $\mathcal{C}$ and continuous parameters
$\boldsymbol{\theta}$, prepares the state
$\ket{\psi(\boldsymbol{\theta},\mathcal{C})}
=U(\boldsymbol{\theta},\mathcal{C})\ket{0}^{\otimes N}$.
The parameters are optimized by minimizing an objective such as
$C(\boldsymbol{\theta},\mathcal{C})
=\bra{\psi(\boldsymbol{\theta},\mathcal{C})}
H\ket{\psi(\boldsymbol{\theta},\mathcal{C})}$ utilizing powerful classical optimizers accessed through SciPy~\cite{virtanen2020scipy}. The structure of a PQC strongly influences the accuracy, trainability, and hardware cost of a VQAs. Two broad PQC design strategies are commonly used in literature. Problem-inspired ans\"atze encode prior knowledge of the target problem, for
example by constructing parameterized circuits from Trotterized evolutions
under terms of the problem Hamiltonian, as in the Hamiltonian variational ansatz~\cite{wecker2015progress}, or from physically motivated excitation operators, as in unitary coupled-cluster~\cite{anand2022quantum} and adaptive VQE methods~\cite{grimsley2019adaptive}. In contrast, hardware-efficient architectures consist of layers of native single- and two-qubit operations arranged to match device connectivity and minimize compilation overhead. The architectures are guided either by problem structure~\cite{peruzzo2014variational} or hardware constraints~\cite{kandala2017hardware}.
\begin{figure}[t]
    \centering
    \includegraphics[width=\linewidth]{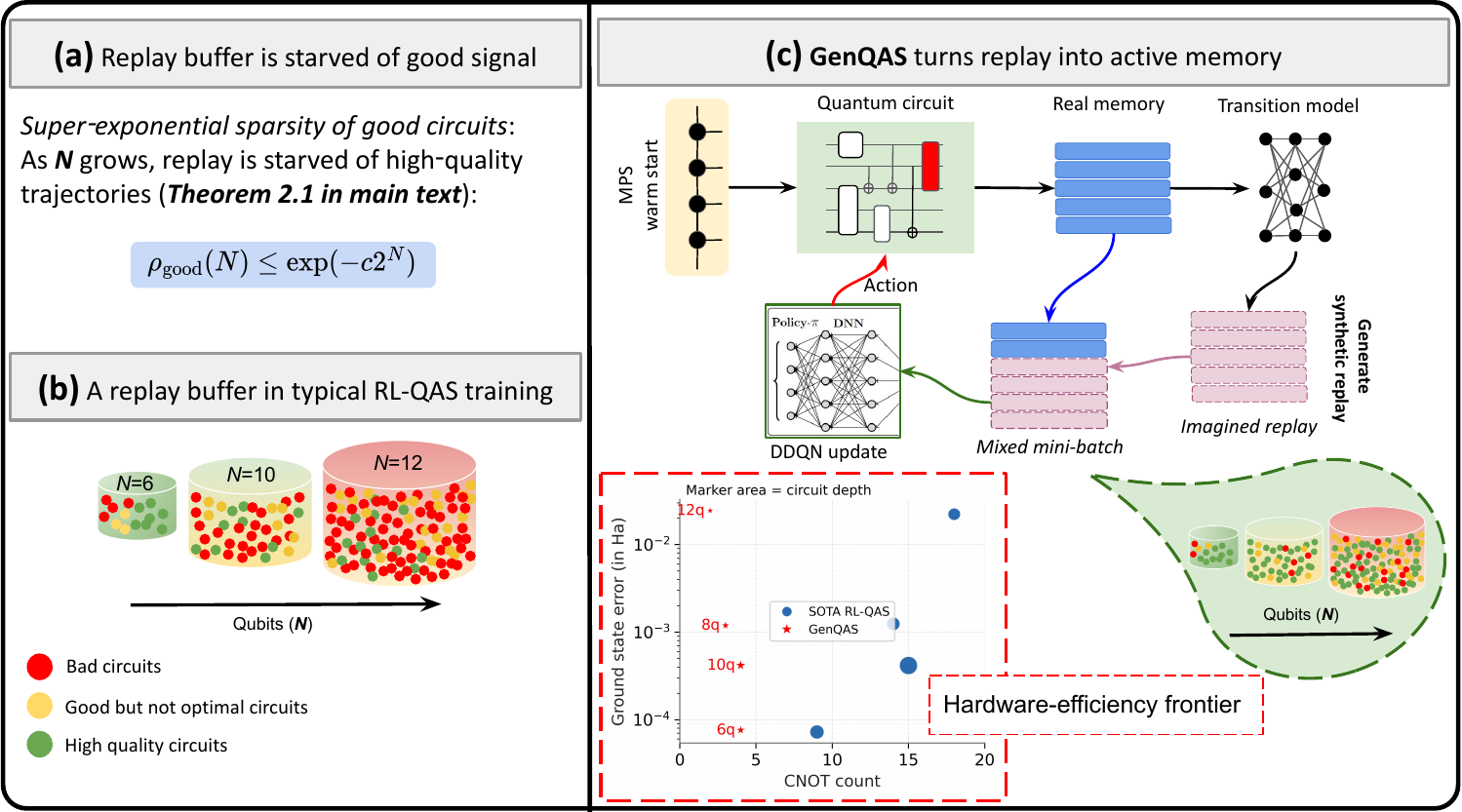}
\caption{\small\textbf{GenQAS converts passive replay into active experience
generation.}
\textbf{a}, Under weakly informed exploration, high-quality circuit
trajectories become increasingly sparse as the system size grows.
\textbf{b}, Uniform and prioritized replay can only resample previously
observed transitions, so the density of useful experience declines as
the search space expands.
\textbf{c}, GenQAS starts from an MPS-derived circuit, retains real
transitions, learns a local transition model and generates synthetic
transitions on the fly. Real and generated samples are mixed during
Double-DQN updates, supporting the search for low-error circuits with
fewer entangling gates.}
    \label{fig:main_framework}
\end{figure}

Although such fixed PQC architecutre can be effective, they may be either insufficiently expressive to access the relevant region of Hilbert space and often stuck in the local minima. To address these limitations a third way was introduced through quantum architecture search (QAS)~\cite{zhang2022differentiable,du2022quantum,kuo2021quantum}. The core mechanism of QAS is to automatically constructing PQC architectures from a prescribed gate set, enabling the circuit structure to adapt to the target task and hardware constraints.

Reinforcement learning (RL) has emerged as a versatile framework for quantum technologies, including the adaptive control of quantum error correction, where an agent uses error-detection events to continuously adjust hardware
control parameters during computation~\cite{sivak2026reinforcement}. RL also provides a natural formulation for automated quantum circuit design and optimization. In the early fault tolerant setting, AlphaTensor-Quantum~\cite{ruiz2025alphatensor}
formulates the optimization of Clifford+$T$ circuits as a symmetric tensor-decomposition problem and uses RL to search for low-rank decompositions corresponding to circuits with reduced $T$-count. In VQAs, RL can instead be used for QAS, in which an agent sequentially constructs a parameterized quantum circuit by selecting gate types and their locations~\cite{ostaszewski2021reinforcement,
kuo2021quantum,fosel2021quantum,moro2021quantum,zhang2020topological}.

This flexibility creates a rapidly expanding decision space. Let
$\mathcal{G}_1$ and $\mathcal{G}_2$ denote the available single- and two-qubit
gate families, respectively, and let $\mathcal{E}_N$ denote the set of
hardware-allowed two-qubit couplings on an $N$-qubit device. The number of
available actions at each circuit-construction step is
\begin{equation}
\lvert\mathcal{A}_N\rvert =
    \underbrace{\lvert\mathcal{G}_1\rvert N}_{\text{1-qubit}}
    +
    \underbrace{\lvert\mathcal{G}_2\rvert
    \lvert\mathcal{E}_N\rvert}_{\text{2-qubit}}
    + 1,
\end{equation}
where the final term denotes termination. For a fully connected device with
undirected couplings,
$\lvert\mathcal{E}_N\rvert=N(N-1)/2$. Consequently, a construction process
of maximum depth $T$ can encounter up to
$\lvert\mathcal{A}_N\rvert^T$ gate sequences, while the dimension of the
underlying Hilbert space grows as $2^N$. Evaluating each candidate architecture
generally further requires continuous parameter optimization and repeated
objective estimation using a quantum simulator or processor. RL-QAS must
therefore learn from a limited number of costly interactions within a search
space that grows combinatorially with circuit depth and rapidly with system
size. For overviews of RL-based quantum circuit optimization and design, we
refer the reader to Refs.~\cite{kundu2026reinforcement,bukov2026reinforcement}.

Recent methods mitigate different aspects of this scaling challenge. Curriculum-based RL-QAS adapts the difficulty of the search task during
training~\cite{patel2024curriculum}, while training-free approaches use proxy metrics to prioritize candidate architectures before costly variational evaluation~\cite{he2024trainingfree}. One-shot methods, such as QuantumNAS~\cite{wang2022quantumnas}, reuse parameters across candidate architectures and incorporate hardware-aware objectives. These approaches reduce the cost of evaluating or prioritizing candidate circuits, but do not directly augment the transition data available to an RL agent during
sequential architecture construction. TensorRL-QAS~\cite{kundu2025tensorrlqas} addresses a complementary aspect of scalability through physically informed initialization. It uses a matrix product state (MPS), obtained with density matrix renormalization group
calculations, to construct a problem-informed initial circuit that is subsequently refined by a double deep Q-network agent. The MPS warm start places the agent in a more relevant region of circuit space and reduces exploration of clearly unpromising architectures. However, after initialization, the agent still learns only from transitions explicitly
collected through interactions with the environment. GenQAS addresses this remaining limitation by learning a generative transition model from real experience and using it to produce on-demand synthetic transitions for DDQN updates, thereby increasing the effective training signal without additional
quantum-classical circuit evaluations.

Experience replay is commonly used to improve sample efficiency by reusing these interactions. Uniform replay samples stored transitions without regard to their learning value, whereas prioritized experience replay increases the sampling frequency of transitions associated with a larger temporal-difference error or another relevance signal~\cite{mnih2013playing,schaul2015prioritized}. Prioritization can therefore make better use of a useful transition after it has been observed. Nevertheless, both approaches remain passive with respect to the set of experiences available for learning: they can change how frequently an observed transition is reused, but they cannot create informative experience that was never collected. As the circuit search space expands, useful trajectories constitute an increasingly small fraction of the agent's interaction history, and the learning signal becomes diluted by a much larger number of low-value alternatives. We refer to this decline in the density of informative experience as \emph{sample starvation}.

Our theoretical analysis formalizes this limitation under the assumptions stated in the Results. For weakly informed exploration over sufficiently expressive circuit ensembles, we show that the measure of near ground state experience decreases exponentially with the dimension of the accessible Hilbert space. This statement characterizes the geometry of random exploration rather than assuming that every learned policy follows the Haar distribution. It nevertheless exposes an important limitation of passive replay during the early stages of search: reweighting previously collected transitions cannot alter the probability that informative transitions were encountered in the first place. One approach to mitigating this limitation is to augment, rather than only reorder, the transitions used in value function updates.

Generative replay provides such a mechanism by learning a model of the local transition dynamics from real interactions~\cite{sutton1991dyna, ha2018worldmodels, janner2019mbpo, feinberg2018modelbased, wang2025prioritized}. Quantum circuit construction is particularly suitable for this approach because the current circuit and the selected gate constrain the subsequent circuit state, the associated reward and whether the construction process terminates. A learned transition model can therefore generate additional training samples around observed state-action pairs without requiring a new quantum-classical evaluation for every generated transition. The main challenge is model bias: synthetic transitions are beneficial only when they remain sufficiently consistent with the real environment to support stable off-policy learning. Generative replay must consequently remain anchored in real experience and control how strongly generated samples influence policy optimization.

Here we introduce \emph{GenQAS}, a tensor network-guided reinforcement learning framework that augments passive experience replay with on-demand synthetic transitions for QAS (Fig.~\ref{fig:main_framework}). GenQAS combines the fixed matrix product state (MPS) initialization and
Double-DQN search backbone of TensorRL-QAS with a learned local transition model. Given a circuit state and gate action, the model predicts the resulting circuit state, reward and termination probability. These predictions define
synthetic one-step transitions that are generated from real state-action seeds and mixed with real transitions during value-function updates.

The MPS initialization and generative replay provide complementary inductive biases. The MPS-derived circuit places the agent in a physically meaningful low-energy region of circuit space, whereas generative replay increases the
amount of training signal extracted from subsequent interactions. Crucially, synthetic transitions are generated on demand and are not stored as replacements for real experience. The replay buffer therefore remains grounded
in environment interactions, while the generation ratio $G_r$ controls the contribution of synthetic experience to each update. This design seeks to improve sample efficiency without requiring additional quantum-classical
parameter optimization calls.

Under the random exploration and Hamiltonian assumptions considered here, we show that near ground state circuit trajectories become rapidly sparse with increasing system size This creates a regime in which resampling observed transitions alone does not increase the support of the observed transition distribution. Within this analytical setting, MPS-seeded generative replay improves the effective density of informative training samples relative to uniform and prioritized replay. The benefit depends on maintaining sufficient agreement between the learned transition model and the underlying QAS environment, thereby motivating the use of real transitions as the persistent replay distribution.

We evaluate GenQAS under matched training budgets on 6-qubit $\mathrm{BeH}_2$, 8-, 10- and 12-qubit $\mathrm{H}_2\mathrm{O}$, and a 15-qubit transverse-field Ising model. At 12 qubits, GenQAS increases final success probability from $12\%$ for the strongest passive baseline to $77\%$, corresponding to a $7.0\times$ improvement. At 15 qubits, it improves the best success probability from $15\%$ with prioritized replay to $21\%$, a 40\% relative increase. Across the tested systems, GenQAS also identifies circuits with competitive energy errors, reduced CNOT counts and circuit depth. In a noisy 6-qubit $\mathrm{BeH}_2$ experiment, transfer of the generative replay mechanism reduces the number of search steps required to obtain the first chemical-accuracy solution from $9053$ to $660$.

We further evaluate GenQAS on a 4-qubit Clifford-synthesis task, in which the agent must construct a circuit matching a target Clifford operation without
variational parameter optimization. In this fully discrete setting, GenQAS exceeds the performance of the passive RL baseline, and the generation frequency provides an additional control over the trade-off between synthetic experience and stable off-policy learning. Together, these results show that model-augmented replay can mitigate sample starvation in RL-QAS and can improve circuit discovery when informative environment interactions are
sparse and expensive.

\section{Results}
\label{sec:results}
Quantum architecture search becomes increasingly difficult as the number of
qubits grows because both the circuit construction space and the Hilbert space
expand rapidly. Our theoretical analysis formalizes one aspect of this
difficulty. Under the weakly informed exploration and Hamiltonian assumptions
specified in Section~\ref{sec:theory}, Theorems~\ref{thm:sparsity}
and~\ref{thm:haar} show that the probability of encountering an
$\epsilon$-accurate low energy circuit through passive exploration decreases
rapidly with system size. Consequently, uniform replay and prioritized
experience replay can only resample a progressively smaller set of useful
transitions.

Lemmas~\ref{lem:mps} and~\ref{lem:fidelity} establish that the MPS derived
warm start provides a controlled initial circuit approximation under the
stated assumptions, while Theorem~\ref{thm:convergence} characterizes the
potential reduction in episode complexity when model generated transitions
remain concentrated in a high fidelity region of the search space. Detailed proofs, assumptions, and derivations are provided in Supplementary
Information Section~\ref{sec:app_proof_theory}. The experiments below test whether the proposed generative replay mechanism produces the corresponding
practical benefits in learning speed, circuit quality, robustness to noise,
and generalization beyond molecular ground state preparation.

We compare GenQAS with TensorRL-QAS using uniform and prioritized replay,
as well as with hardware efficient ansatz and UCCSD baselines where
applicable. Molecular benchmarks include 6 qubit $\behtwo$ and 8, 10, and
12 qubit $\htwoo$. We additionally consider a 15 qubit transverse field Ising
model and a 4 qubit Clifford synthesis task. Within each benchmark, all RL
methods use the same environment, tensor network warm start, action space,
and benchmark specific training budget. They differ only in the replay
mechanism. Full hyperparameters, molecular Hamiltonian specifications, and
simulation details are given in Supplementary Information
Sections~\ref{sec:genqas_implement_details},
\ref{sec:molecular_hamiltonians}, and
\ref{sec:quantum_simulation_details}.

\begin{figure}[t!]
  \centering

  \begin{subfigure}{\linewidth}
    \centering
    \includegraphics[width=\linewidth]{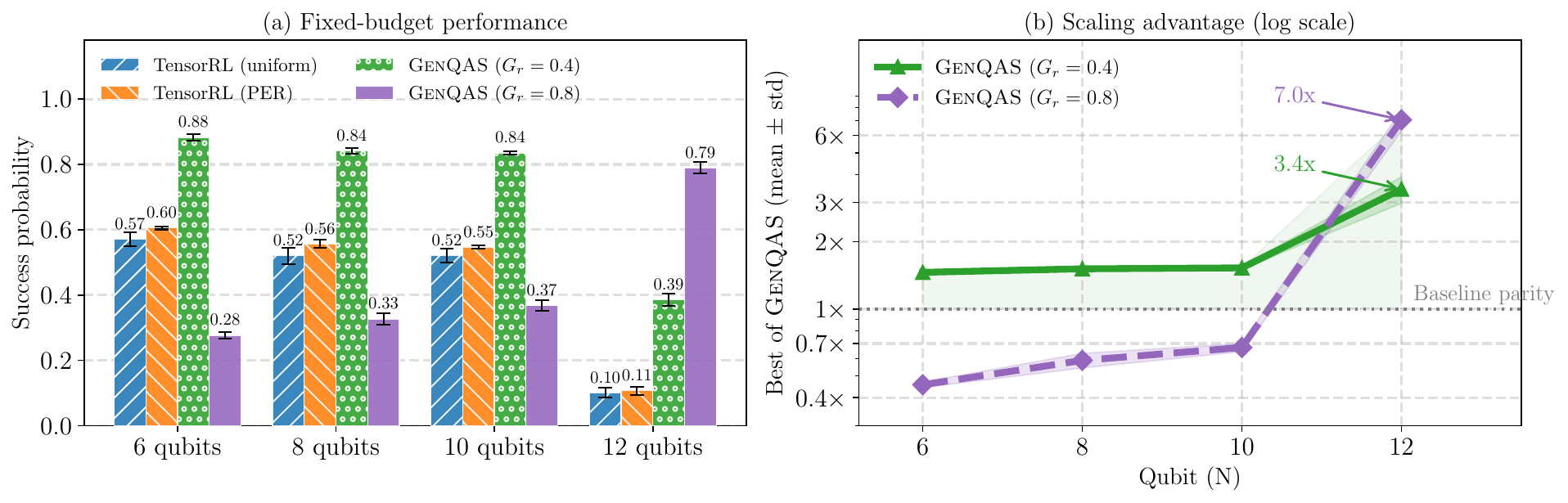}
    \caption{Scaling performance over 5 random initializations of neural network.}
    \label{fig:scaling}
  \end{subfigure}

  \vspace{0.8em}

  \begin{subfigure}{\linewidth}
    \centering
    \includegraphics[width=\linewidth]{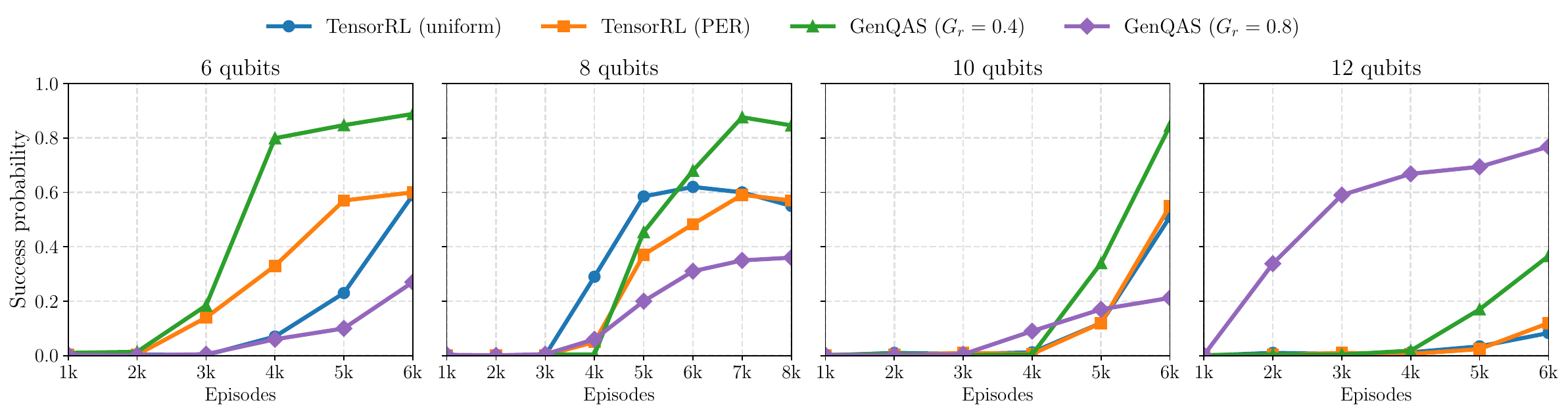}
    \caption{Training dynamics.}
    \label{fig:training}
  \end{subfigure}

  \caption{\small \textbf{Fixed-budget performance and training dynamics of {GenQAS}.}
  \textbf{(a)}~At a fixed budget of $6000$ training episodes, {GenQAS} consistently outperforms TensorRL-QAS with uniform and PER replay across $N \in \{6, 8, 10, 12\}$, while passive baselines degrade sharply with system size. The observed gains increase with qubit number, reaching the strongest advantage at 12 qubits and showing an increasing relative advantage over the evaluated 6 to 12 qubit range, consistent with the proposed signal-density mechanism.
  \textbf{(b)}~Training trajectories for 6, 8, 10 and 12 qubits show that {GenQAS} converges faster and attains higher success probability than TensorRL-QAS baselines. The $G_r = 0.8$ setting is the most robust at larger system sizes, maintaining high success even at 12 qubits.}
  \label{fig:genqas_stacked}
\end{figure}
\subsection{Generative replay improves fixed budget scaling}
\label{sec:scaling}

Theorem~\ref{thm:convergence} predicts that, under the assumptions of the
theoretical analysis, generative replay can reduce the number of environment
interactions required to reach a target accuracy relative to passive replay.
Here, we test the empirical consequence of this prediction under matched
training budgets at each system size: whether the advantage of GenQAS over
passive replay becomes more pronounced as the QAS problem increases from 6 to
12 qubits.

At each DDQN update, GenQAS mixes real transitions from replay memory with
synthetic transitions generated on demand by the learned dynamics model. The
generation ratio $G_r\in[0,1)$ specifies the synthetic fraction of a total
update batch of size $B$:
\begin{equation}
    B_{\mathrm{real}}
    =
    \left\lfloor
    (1-G_r)B
    \right\rfloor,
    \qquad
    B_{\mathrm{syn}}
    =
    B-B_{\mathrm{real}}.
    \label{eq:scaling_generation_ratio}
\end{equation}
Thus, $G_r=0.4$ produces an update batch with approximately $60\%$ real and
$40\%$ synthetic transitions, whereas $G_r=0.8$ produces approximately
$20\%$ real and $80\%$ synthetic transitions. Synthetic samples are generated
from state action seeds drawn from real replay memory and are not stored
persistently.

The theoretical mechanism can be expressed through the useful-transition
density in an update batch. If $\rho_{\mathrm{real}}(N)$ and
$\rho_{\mathrm{syn}}(N)$ denote the densities of useful real and synthetic
transitions, respectively, then the mixed update distribution has density
\begin{equation}
    \rho_{\mathrm{mix}}(N)
    =
    (1-G_r)\rho_{\mathrm{real}}(N)
    +
    G_r\rho_{\mathrm{syn}}(N).
    \label{eq:mixed_signal_density}
\end{equation}
Under the assumptions of the theoretical analysis,
Theorems~\ref{thm:sparsity} and~\ref{thm:haar} show that
$\rho_{\mathrm{real}}(N)$ decreases rapidly under passive exploration.
Lemma~\ref{lem:fidelity} establishes a fidelity bound for the MPS derived
warm start, while Theorem~\ref{thm:convergence} characterizes the conditions
under which generative replay can maintain a nonvanishing useful-transition
contribution to the mixed update distribution. The full signal-density
derivation and its assumptions are provided in Supplementary Information
Section~\ref{sec:supp_signal_density}.

Figure~\ref{fig:scaling} reports the final success probability across 6, 8,
10, and 12 qubits, averaged over five random neural-network initializations.
At 6 qubits, TensorRL-QAS with uniform and prioritized replay reaches success
probabilities of approximately $0.57$ and $0.60$, respectively. GenQAS with
$G_r=0.4$ reaches $0.88$, whereas the more synthetic configuration,
$G_r=0.8$, reaches $0.28$. At 8 qubits, the passive baselines reach
approximately $0.54$ and $0.55$, while GenQAS reaches $0.84$ for $G_r=0.4$
and $0.34$ for $G_r=0.8$.

The same qualitative pattern persists at 10 qubits. The uniform and
prioritized TensorRL-QAS baselines attain final success probabilities of
approximately $0.51$ and $0.55$, respectively. GenQAS with $G_r=0.4$ reaches
$0.84$, corresponding to an approximately $1.5\times$ improvement over the
strongest passive baseline, whereas GenQAS with $G_r=0.8$ reaches $0.37$.
These 6 to 10 qubit results show that an intermediate synthetic fraction
provides the strongest performance when real transitions remain sufficiently
available to anchor value learning.

At 12 qubits, the passive replay baselines deteriorate sharply, reaching
success probabilities of approximately $0.10$ and $0.11$. In contrast,
GenQAS with $G_r=0.4$ attains $0.38$, corresponding to an approximately
$3.4\times$ improvement over the strongest passive baseline. The
$G_r=0.8$ configuration reaches $0.77$, corresponding to an approximately
$7.0\times$ improvement. Thus, the generation ratio that is suboptimal at
smaller system sizes becomes the strongest configuration when the useful
real-transition density is most severely depleted.

The fixed budget advantage over the evaluated instances shown in Figure~\ref{fig:scaling} is therefore not
strictly monotonic in $G_r$. At 6, 8, and 10 qubits, $G_r=0.4$ provides the
highest final success probability, indicating that a substantial real-data
component remains beneficial. At 12 qubits, however, $G_r=0.8$ becomes
superior, consistent with the hypothesis that a larger synthetic contribution
is increasingly valuable once passive replay enters a stronger
sample-starvation regime. Although the four investigated system sizes are
insufficient to establish an asymptotic scaling law empirically, the
increasing relative advantage of GenQAS is consistent with the signal-density
mechanism in Eq.~\eqref{eq:mixed_signal_density}.

Figure~\ref{fig:training} shows the corresponding learning trajectories.
At 6 qubits, GenQAS with $G_r=0.4$ rapidly improves after approximately
3000 episodes and reaches the highest final success probability. At 8 and
10 qubits, the same configuration again converges faster and reaches a higher
final plateau than uniform replay, prioritized replay, and the $G_r=0.8$
configuration. At 12 qubits, the learning dynamics change qualitatively:
the passive baselines remain near zero success probability throughout most of
training, while GenQAS with $G_r=0.8$ steadily improves and reaches the
highest final success probability. These trajectories show that the preferred
real to synthetic replay mixture depends on problem size, with higher
generation ratios becoming advantageous in the largest QAS instance examined.

\subsection{Generative replay improves noisy replay buffer transfer}
\label{sec:noisy_transfer}

We next evaluate whether replay information collected in a noiseless setting
remains useful when the same QAS task is evaluated under gate noise. We
consider the 6-qubit BeH$_2$ benchmark and follow the replay buffer transfer
setting~\cite{kundu2026replay}. TensorRL-QAS and GenQAS are first
trained in the noiseless environment using the same reward construction and
training protocol as in the molecular experiments. The resulting replay
memories are subsequently used to initialize training in a noisy environment.
Only the replay-buffer information is transferred, network weights and an
additional $\epsilon$-greedy pretraining phase are not transferred.

The noisy environment applies a depolarizing channel after every gate. Each single-qubit gate is followed by depolarizing noise of strength
$p_1=0.01$, and each two-qubit gate is followed by depolarizing noise of
strength $p_2=0.05$. TensorRL-QAS uses a transferred uniform replay buffer,
whereas GenQAS uses the transferred prioritized generative replay mechanism
trained with $G_r=0.4$. 
\begin{figure}[h!]
  \centering
  \begin{subfigure}{0.70\linewidth}
    \centering
    \includegraphics[width=\linewidth]{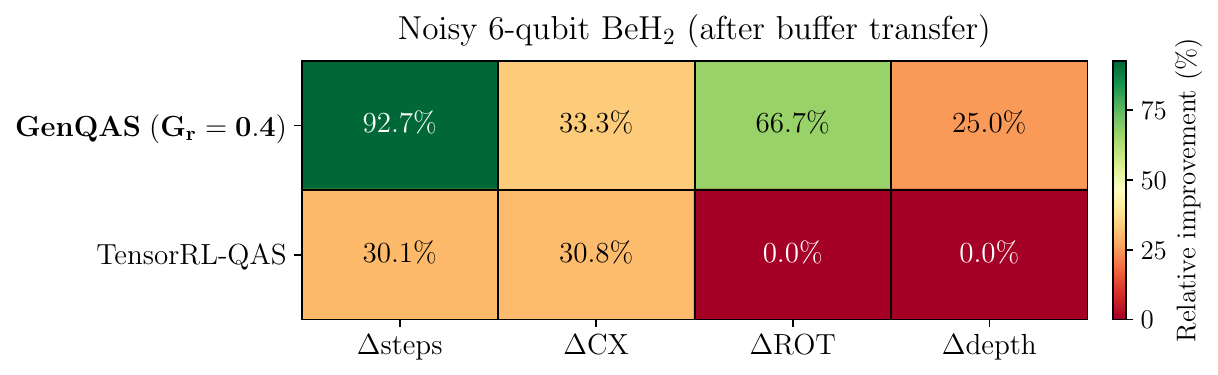}
    \caption{Relative improvement after noiseless replay-buffer transfer.}
    \label{fig:noisy_6q_heatmap}
  \end{subfigure}

  \vspace{0.4em}

  \begin{subfigure}{0.90\linewidth}
    \centering
    \small
    \begin{tabular}{lcccc}
      \toprule
      Run & Steps to chemical accuracy & CX & ROT & Depth \\
      \midrule
      TensorRL-QAS~\cite{kundu2025tensorrlqas}
        & 13\,623 & 13 & 4 & 10 \\
      TensorRL-QAS with transfer
        & 9\,520 & 9 & 4 & 10 \\
      GenQAS
        & 9\,053 & 9 & 9 & 8 \\
      GenQAS with transfer
        & 660 & 6 & 3 & 6 \\
      \bottomrule
    \end{tabular}
    \caption{Raw metrics for the noisy 6-qubit BeH$_2$ task.}
    \label{fig:noisy_6q_table}
  \end{subfigure}

  \caption{\small \textbf{Replay buffer transfer in the noisy 6-qubit $\behtwo$ experiment}. A depolarizing channel with strength $p_1=0.01$ is applied
  after each single-qubit gate and a channel with strength $p_2=0.05$ is applied after each two-qubit gate. (a) Relative improvement in steps to chemical accuracy, CNOT count, rotation count, and circuit depth after transfer of replay information collected in the noiseles environment.
  (b) Corresponding raw metrics for TensorRL-QAS and GenQAS with and without
  replay-buffer transfer.}
  \label{fig:noisy_6q}
\end{figure}

Figure~\ref{fig:noisy_6q} shows that transferred replay information improves
both methods, but the improvement is substantially larger for GenQAS. For
TensorRL-QAS, transfer of the uniform replay buffer reduces the steps to the
first chemical-accuracy solution from 13623 to 9520, a $30.1\%$
reduction. The transferred buffer also reduces the CNOT count from 13 to 9,
a $30.8\%$ reduction, but leaves the rotation count and circuit depth
unchanged at 4 and 10, respectively.

For GenQAS, transfer of the generative replay mechanism reduces the steps to
chemical accuracy from 9053 to 660, corresponding to a $92.7\%$ reduction.
The transferred GenQAS configuration also decreases the CNOT count from 9 to
6, the rotation count from 9 to 3, and the circuit depth from 8 to 6. These
correspond to relative reductions of $33.3\%$, $66.7\%$, and $25.0\%$,
respectively.

The result shows that, in this noiseless to noisy transfer experiment, the
GenQAS replay representation retains more useful information for subsequent
noisy architecture search than a passive uniform replay buffer. In
particular, the transferred GenQAS configuration reaches chemical accuracy
with fewer search steps while also identifying a more compact circuit. This
finding is consistent with the hypothesis that model-generated replay can
retain locally useful circuit-transition structure across a moderate change in
the evaluation environment. It does not establish noise-model independence,
but it demonstrates a substantial transfer benefit under the depolarizing
noise setting considered here.

\subsection{Generative replay diagnostics reveal a generation-ratio trade-off}
\label{sec:buffer-analysis}

The fixed-budget results show that generative replay can improve policy
performance, particularly as useful real transitions become sparse. However,
synthetic replay is beneficial only when the learned transition model remains
sufficiently accurate for stable value learning. We therefore examine the
generator losses and weighted DDQN Huber losses at 6 and 10 qubits for
generation ratios $G_r\in\{0.4,0.6,0.8\}$. These diagnostics do not directly
measure the density of low-energy synthetic transitions. Instead, they assess
whether the learned dynamics model and value function remain numerically
well behaved under different fractions of generated replay.

\begin{figure}[t!]
  \centering
  \includegraphics[width=\linewidth]{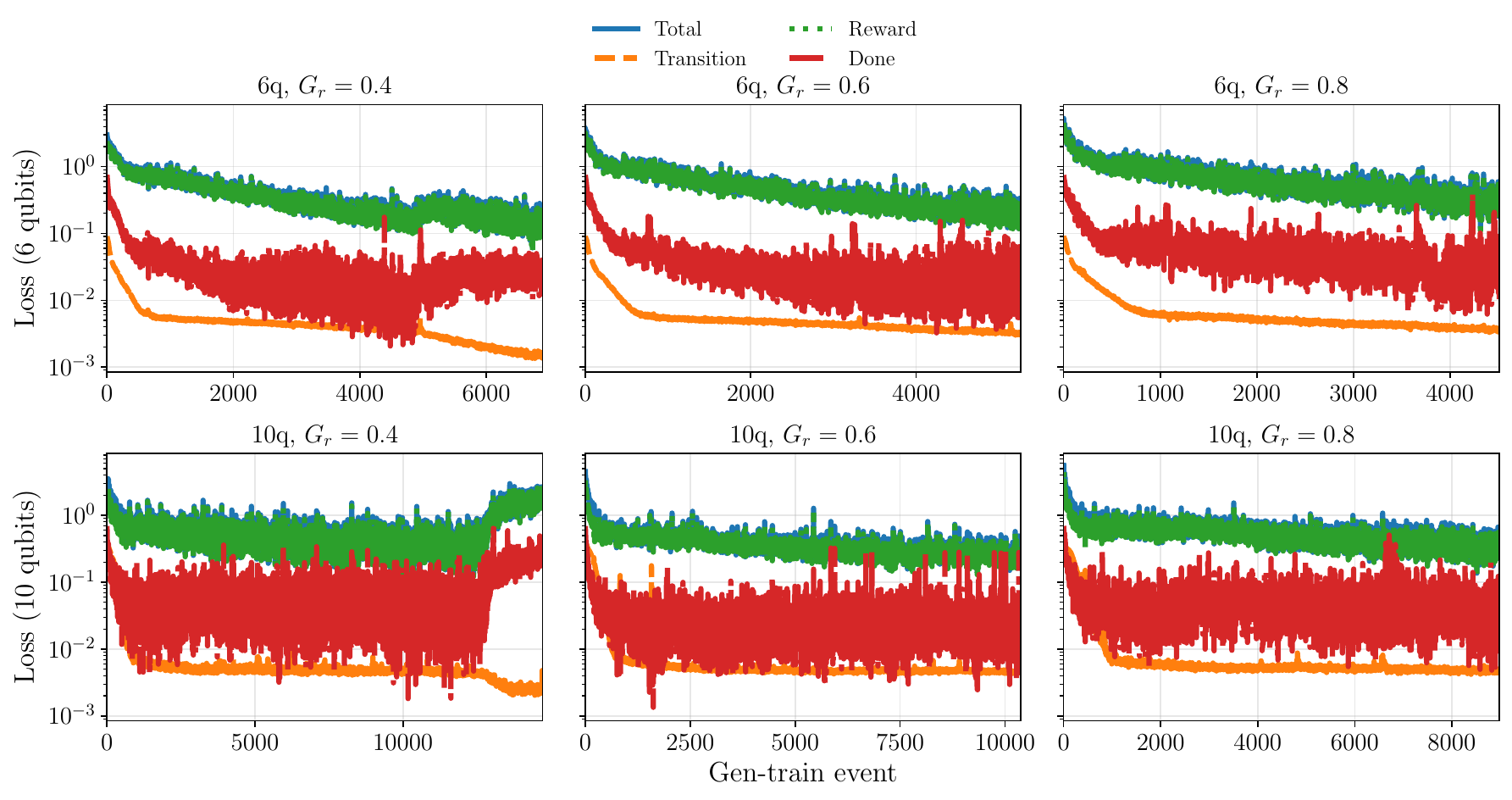}
  \caption{\small \textbf{Generator loss decomposition at 6 and 10 qubits}.
  Total, transition, reward, and termination losses of the learned dynamics
  model for generation ratios $G_r\in\{0.4,0.6,0.8\}$. The transition loss
  decreases or remains low across the investigated settings. At 10 qubits,
  the $G_r=0.4$ configuration exhibits a late increase in total loss, whereas
  the $G_r=0.6$ and $G_r=0.8$ configurations remain more bounded over their
  respective training intervals.}
  \label{fig:gen-losses}
\end{figure}
\begin{figure}[h!]
    \centering
    \includegraphics[width=0.8\linewidth]{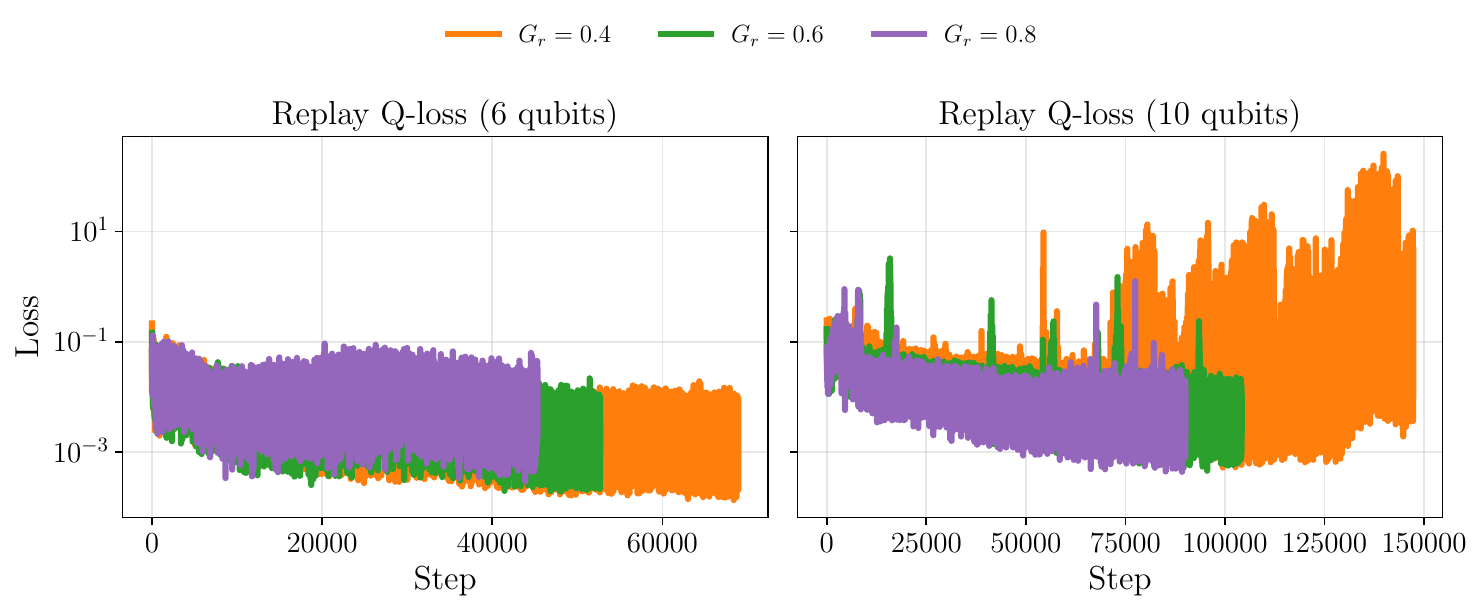}
    \caption{\small \textbf{Weighted DDQN Huber loss under different generation
    ratios}. Replay Q-learning loss over training steps at 6 and 10 qubits for
    $G_r\in\{0.4,0.6,0.8\}$. At 6 qubits, all configurations remain in a low
    bounded range after the initial training stage. At 10 qubits,
    $G_r=0.4$ exhibits late high-variance loss spikes, whereas the
    $G_r=0.6$ and $G_r=0.8$ configurations remain comparatively more stable
    over the displayed training intervals.}
    \label{fig:q-loss}
\end{figure}

Figure~\ref{fig:gen-losses} decomposes the dynamics-model objective into
transition, reward, and termination components. At 6 qubits, the transition
loss decreases substantially for all generation ratios, reaching values near
$10^{-3}$, while the total loss remains bounded. The reward and termination
components exhibit stochastic fluctuations, but no sustained divergence is
observed. These results show that the local transition model can be trained
stably while generated replay constitutes between $40\%$ and $80\%$ of each
DDQN update batch.

At 10 qubits, the dependence on $G_r$ is more pronounced. For $G_r=0.6$ and
$G_r=0.8$, the transition loss remains low and the total generator loss stays
within a relatively bounded range. In contrast, the $G_r=0.4$ configuration
shows a late increase in the total loss, driven primarily by the reward and
termination components. This behaviour indicates that generator stability is
not determined by the generation ratio alone; it also depends on the evolving
state distribution, reward landscape, and training horizon of the underlying
QAS task.

Figure~\ref{fig:q-loss} reports the weighted DDQN Huber loss computed from
mixed batches of real and synthetic transitions. At 6 qubits, the losses for
all three generation ratios decrease after the initial learning regime and
remain within a low range. The $G_r=0.8$ setting displays moderately larger
early fluctuations, but its loss remains bounded over the reported training
interval.

At 10 qubits, the $G_r=0.4$ configuration develops large late-stage Q-loss
spikes. By comparison, the $G_r=0.6$ and $G_r=0.8$ configurations retain a
lower and more bounded loss profile, apart from isolated fluctuations. The
loss diagnostics therefore identify a generation-ratio trade-off: a moderate
synthetic fraction can yield strong final success probability, as observed in
the 10-qubit scaling experiment, while a larger synthetic fraction can provide
a more stable critic-loss trajectory in the later stages of training.

Together, Figures~\ref{fig:gen-losses} and~\ref{fig:q-loss} show that
on-demand synthetic replay can be incorporated into DDQN training without
systematic divergence across the investigated settings. They also show that
the relative stability of different generation ratios changes with system
size. This motivates treating $G_r$ as a task-dependent control parameter,
rather than assuming that a fixed real-to-synthetic replay mixture is optimal
across all QAS instances.

\begin{table*}[h!]
\centering
\fontsize{8}{9.6}\selectfont
\caption{\small\textbf{Circuit statistics for BeH$_2$ (6 qubits) and H$_2$O (8, 10, and 12 qubits).} For each configuration we report two-qubit gate count (\texttt{CNOT}), single-qubit rotation count (\texttt{ROT}),
circuit depth, and error for the best seed. GenQAS consistently attains comparable or lower error than TensorRL-QAS, HEA, and UCCSD while using substantially fewer two-qubit gates, fewer rotations, or reduced
depth. HEA errors are taken from the best of 10 random seeds per ansatz depth (1, 2, 3 layers); UCCSD statistics are from a single deterministic run.}
\label{tab:pgr_compare}
\begin{tabular}{llcccc}
\toprule
System & Method & \texttt{CNOT} & \texttt{ROT} & Depth & Error \\
\midrule
\multirow{10}{*}{6-qubit BeH$_2$}
  & TensorRL-QAS (Uniform)           &  9   &  9   &  9   & $7.22\times10^{-5}$ \\
  & TensorRL-QAS (PER)               & 16   &  3   & 14   & $7.89\times10^{-5}$ \\
  & \textbf{GenQAS} ($G_r=0.2$)      & 14   &  5   & 14   & $7.62\times10^{-5}$ \\
  & \textbf{GenQAS} ($G_r=0.4$)      &  5   & \textbf{2} &  6   & $7.18\times10^{-5}$ \\
  & \textbf{GenQAS} ($G_r=0.6$)      &  5   & 10   & \textbf{5} & \textbf{$2.37\times10^{-5}$} \\
  & \textbf{GenQAS} ($G_r=0.8$)      & \textbf{4} & 12 &  6   & $7.59\times10^{-5}$ \\
  & HEA (1-layer)                    &  6   & 24   & 10   & $5.92\times10^{-3}$ \\
  & HEA (2-layer)                    & 12   & 36   & 18   & $1.98\times10^{-3}$ \\
  & HEA (3-layer)                    & 18   & 48   & 26   & $3.37\times10^{-4}$ \\
  & UCCSD                            & 272  & 620  & 475  & $5.92\times10^{-3}$ \\
\midrule
\multirow{10}{*}{8-qubit H$_2$O}
  & TensorRL-QAS (Uniform)           & 14   &  6   &  8   & $1.24\times10^{-3}$ \\
  & TensorRL-QAS (PER)               & 14   &  2   & 11   & $1.26\times10^{-3}$ \\
  & \textbf{GenQAS} ($G_r=0.2$)      & 12   &  2   & 10   & $2.22\times10^{-3}$ \\
  & \textbf{GenQAS} ($G_r=0.4$)      &  3   & \textbf{2} & \textbf{4} & \textbf{$1.19\times10^{-3}$} \\
  & \textbf{GenQAS} ($G_r=0.6$)      &  7   & 10   & 10   & $1.22\times10^{-3}$ \\
  & \textbf{GenQAS} ($G_r=0.8$)      & \textbf{2} & 13 &  8   & $1.20\times10^{-3}$ \\
  & HEA (1-layer)                    &  8   & 32   & 12   & $2.63\times10^{-3}$ \\
  & HEA (2-layer)                    & 16   & 48   & 22   & $2.63\times10^{-3}$ \\
  & HEA (3-layer)                    & 24   & 64   & 32   & $2.96\times10^{-3}$ \\
  & UCCSD                            & 1312 & 2596 & 2148 & $2.63\times10^{-3}$ \\
\midrule
\multirow{10}{*}{10-qubit H$_2$O}
  & TensorRL-QAS (Uniform)           & 15   & 17   & 17   & $4.15\times10^{-4}$ \\
  & TensorRL-QAS (PER)               & 23   & 14   & 19   & $4.17\times10^{-4}$ \\
  & \textbf{GenQAS} ($G_r=0.2$)      & \textbf{4} & \textbf{5} & \textbf{7} & $4.17\times10^{-4}$ \\
  & \textbf{GenQAS} ($G_r=0.4$)      & 14   & 15   & 20   & \textbf{$3.92\times10^{-4}$} \\
  & \textbf{GenQAS} ($G_r=0.6$)      & 19   & 19   & 20   & $3.93\times10^{-4}$ \\
  & \textbf{GenQAS} ($G_r=0.8$)      & 26   & 36   & 24   & $4.32\times10^{-4}$ \\
  & HEA (1-layer)                    & 10   & 40   & 14   & $4.69\times10^{-3}$ \\
  & HEA (2-layer)                    & 20   & 60   & 26   & $4.84\times10^{-3}$ \\
  & HEA (3-layer)                    & 30   & 80   & 38   & $1.14\times10^{-2}$ \\
  & UCCSD                            & 3440 & 5932 & 5336 & $4.69\times10^{-3}$ \\
\midrule
\multirow{10}{*}{12-qubit H$_2$O}
  & TensorRL-QAS (Uniform)           & 18   &  2   &  6   & \textbf{$2.22\times10^{-2}$} \\
  & TensorRL-QAS (PER)               & 17   &  2   & 11   & $2.48\times10^{-2}$ \\
  & \textbf{GenQAS} ($G_r=0.2$)      & 14   &  2   & 10   & $2.33\times10^{-2}$ \\
  & \textbf{GenQAS} ($G_r=0.4$)      & 18   & \textbf{2} &  4   & \textbf{$2.28\times10^{-2}$} \\
  & \textbf{GenQAS} ($G_r=0.6$)      &  3   & 10   & 10   & $2.45\times10^{-2}$ \\
  & \textbf{GenQAS} ($G_r=0.8$)      & \textbf{2} &  4   & \textbf{2} & $2.44\times10^{-2}$ \\
  & HEA (1-layer)                    & 12   & 48   & 16   & $2.52\times10^{-2}$ \\
  & HEA (2-layer)                    & 24   & 72   & 30   & $3.29\times10^{-2}$ \\
  & HEA (3-layer)                    & 36   & 96   & 44   & $6.31\times10^{-2}$ \\
  & UCCSD                            & 6976 & 10628 & 10356 & $2.51\times10^{-2}$ \\
\bottomrule
\end{tabular}
\end{table*}

\subsection{GenQAS identifies compact circuits at competitive error}
\label{sec:baseline_comparison}

Table~\ref{tab:pgr_compare} compares GenQAS with TensorRL-QAS using uniform and prioritized experience replay, as well as with fixed hardware-efficient ansatz and UCCSD reference circuits. The RL based methods share the same MPS derived warm start, DDQN backbone, environment, and action space. Thus, the comparison between GenQAS and TensorRL-QAS isolates the effect of replacing passive replay with prioritized generative replay. For each RL configuration, the table reports the most accurate circuit among the evaluated random seeds

Across the molecular benchmarks, GenQAS identifies circuits on a competitive accuracy-resource frontier. At 6 qubits, GenQAS with $G_r=0.6$ achieves the lowest reported error, $2.37\times10^{-5}$, at depth 5 using 5 CNOTs. At 8 qubits, GenQAS with $G_r=0.4$ reaches an error of $1.19\times10^{-3}$ using 3 CNOTs, 2 rotations, and depth 4, compared with 14 CNOTs, 6 rotations, and depth 8 for TensorRL-QAS with uniform replay. At 10 qubits, GenQAS with $G_r=0.2$ matches the error of prioritized replay while using 4 CNOTs, 5 rotations, and depth 7.

The 12 qubit H$_2$O task illustrates the accuracy-resource trade-off. Uniform TensorRL-QAS reaches the lowest reported error,
$2.22\times10^{-2}$, whereas GenQAS with $G_r=0.8$ identifies the most
compact circuit, using 2 CNOTs, 4 rotations, and depth 2 with error
$2.44\times10^{-2}$. GenQAS with $G_r=0.4$ achieves an intermediate
solution with error $2.28\times10^{-2}$ and depth 4. Thus, generative replay does not uniformly dominate passive replay for every individual metric; instead, it provides a practical means of navigating the trade-off between accuracy and circuit resources.

The fixed HEA and UCCSD circuits provide reference points for the size of the
resource savings. At 12 qubits, UCCSD requires 6976 CNOTs, 10628
single-qubit rotations, and depth 10356, whereas the GenQAS circuits use at most tens of gates and depth at most 10. HEA circuits are substantially more compact than UCCSD but remain less gate efficient than the compact GenQAS solutions at comparable error. These comparisons indicate that generative replay can support the discovery of low resource circuit architectures under the same RL interaction budget.

\subsection{Generative replay generalizes across quantum optimization tasks}
\label{sec:generalization}

The molecular benchmarks establish that generative replay can improve fixed-budget learning and identify compact circuits for electronic structure Hamiltonians. We next test whether this behaviour extends beyond molecular
ground state preparation. We consider a 15-qubit transverse field Ising model (TFIM), which preserves the variational ground state objective but changes the physical problem class, and a 4-qubit Clifford synthesis task, which replaces
variational energy minimization with exact discrete target matching.

\paragraph{Spin-model ground state preparation}

For the 15 qubit TFIM, we evaluate GenQAS using the TensorRL-QAS backbone under a matched training budget. Figure~\ref{fig:tfim15_plot_table}(a) shows that GenQAS attains equal or higher final success probability than TensorRL with uniform or prioritized replay for the tested generation ratios. The final energy errors remain in the narrow range from $1.53\times10^{-2}$ to
$1.54\times10^{-2}$ across all methods, indicating that the replay modification does not degrade final solution quality in this benchmark.

The generation ratio determines the resource profile of the resulting circuits. GenQAS with $G_r=0.4$ matches the error and CNOT count of the uniform TensorRL baseline while reducing the rotation count from 8 to 6. GenQAS with $G_r=0.8$ identifies a circuit with only 3 CNOTs at comparable error, although it uses 13 single-qubit rotations. Thus, $G_r$ provides a
\begin{figure}[h!]
\centering

\begin{subfigure}[t]{0.4\linewidth}
    \centering
    \includegraphics[width=\linewidth]{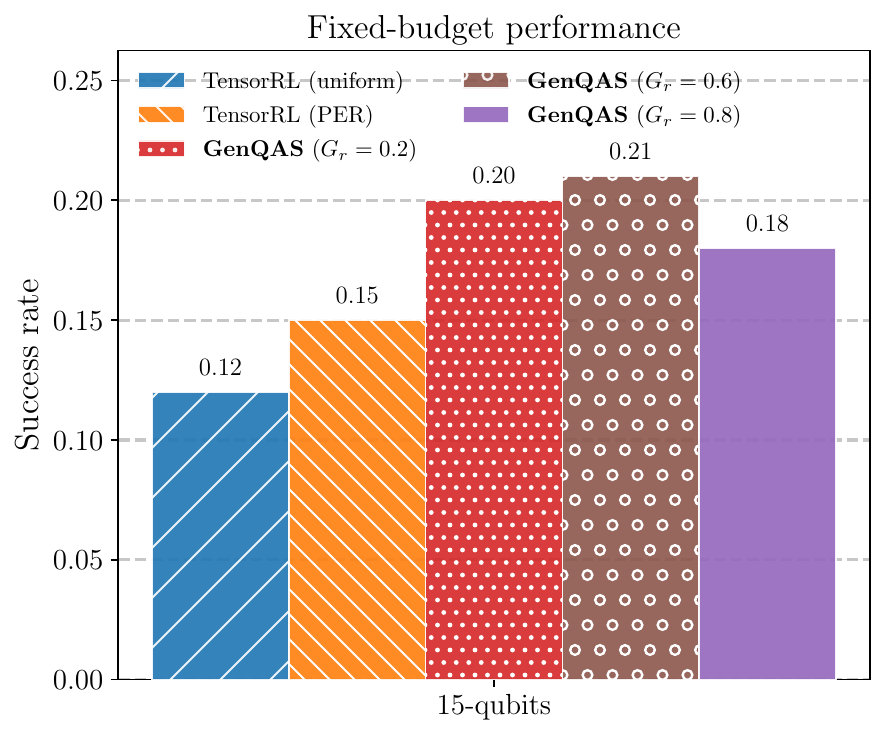}
    \caption{Fixed-budget performance.}
    \label{fig:tfim15_final_success}
\end{subfigure}
\hfill
\begin{subfigure}[t]{0.55\linewidth}
    \centering
    \vspace{-4cm}
    \begin{minipage}[t]{\linewidth}
        \fontsize{7}{9.6}\selectfont
        \centering
        \begin{tabular}{lccc}
            \toprule
            Method & Error & CNOT & ROT \\
            \midrule
            TensorRL (uniform) & $1.53\times10^{-2}$ & 10 & 8 \\
            TensorRL (PER)     & $1.54\times10^{-2}$ & 10 & 8 \\
            \textbf{GenQAS} ($G_r=0.2$) & $1.53\times10^{-2}$ &  7 & 10 \\
            \textbf{GenQAS} ($G_r=0.4$) & $1.53\times10^{-2}$ & 10 & 6 \\
            \textbf{GenQAS} ($G_r=0.6$) & $1.54\times10^{-2}$ & 14 & 6 \\
            \textbf{GenQAS} ($G_r=0.8$) & $1.53\times10^{-2}$ &  3 & 13 \\
            \bottomrule
        \end{tabular}
        \caption{Circuit statistics.}
        \label{tab:tfim15_error_table}
    \end{minipage}
\end{subfigure}

\caption{15-qubit TFIM performance: fixed-budget success (a) and best circuit statistics (b).}
\label{fig:tfim15_plot_table}
\end{figure}
practical control parameter for navigating the trade-off between entangling
gates, single-qubit rotations, and learning behaviour in a non-molecular
ground-state problem.

\paragraph{Discrete clifford synthesis}

We further evaluate GenQAS on a 4-qubit Clifford synthesis task, in which the agent sequentially constructs a circuit that exactly matches a target Clifford operation using a discrete gate library $\mathcal{A}_{\mathrm{Cliff}}
    =
    \left\{
        \texttt{H}_q,\,
        \texttt{S}_q,\,
        \texttt{S}_q^{\dagger},\,
        \texttt{CNOT}_{c,t}
    \right\}$,
    where
    $q\in\{0,\ldots,N-1\}$, and
    $c\neq t$. Unlike the molecular and TFIM benchmarks, this task requires neither variational parameter optimization nor energy expectation-value estimation. Rewards are determined by exact target matching, allowing us to test generative replay in a fully discrete circuit construction setting. The target generation procedure, action space construction, reward design, and evaluation protocol are provided in Supplementary Information Section~\ref{sec:supp_clifford_synthesis_exp}.
\begin{figure}[h!]
    \centering
    \includegraphics[width=0.7\linewidth]{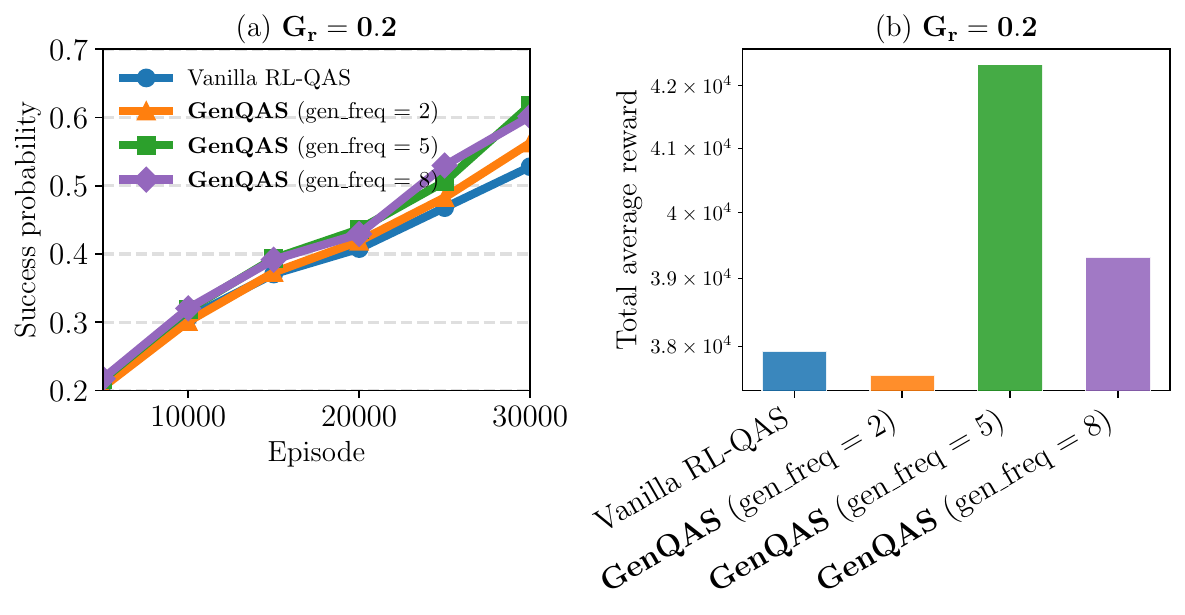}
    \caption{\small{4-qubit clifford synthesis. Having 4-qubit clifford synthesis a relatively easier problem to solve here we wanted to show that a small generation ratio ($G_r$) makes a big difference in performance. As we do not have the variational quantum-classical framework in this experiment, this is the exact place where we can analyse the generation frequency by the \texttt{conditional generative model} implemented inside the genQAS.}}
    \label{fig:clifford_analysis}
\end{figure}

Figure~\ref{fig:clifford_analysis}(a) shows that GenQAS with $G_r=0.2$
matches or modestly exceeds the success probability of the passive RL-QAS baseline. Figure~\ref{fig:clifford_analysis}(b) further shows that the generator update frequency affects learning performance. The intermediate configuration $\texttt{gen\_freq}=5$ achieves the highest average cumulative reward under the shared training budget, whereas very infrequent and very frequent generation provide smaller gains.

Together, the TFIM and Clifford results show that the benefit of generative replay is not restricted to molecular Hamiltonians. It can also improve sequential circuit construction in a spin-model ground-state task and in a fully discrete exact-synthesis task. At the same time, both experiments show that the generation ratio and update frequency are task-dependent hyperparameters rather than universally optimal constants.

\section{Discussion}
\label{sec:discussion}

This work addresses sample starvation in reinforcement learning for quantum architecture search. Under the random-exploration and Hamiltonian assumptions of our theoretical analysis, Theorems~\ref{thm:sparsity} and~\ref{thm:haar} show that the probability of encountering an $\epsilon$-accurate low-energy circuit decreases rapidly as the number of qubits increases. In this regime, uniform replay and prioritized experience replay can change how often observed transitions are reused, but cannot alter the probability that informative transitions were collected in the first place. GenQAS addresses this limitation by combining a fixed MPS derived warm
start with a learned local transition model that generates additional
on-demand replay transitions from real state-action seeds. Under the
conditions stated in Lemma~\ref{lem:fidelity} and Theorem~\ref{thm:convergence}, this mechanism can maintain a nonvanishing useful transition contribution to the mixed replay distribution. The complete signal-density derivation and its assumptions are provided in Supplementary Information Section~\ref{sec:supp_signal_density}.

The empirical results in Section~\ref{sec:results} are consistent with this mechanism. At 12 qubits, GenQAS increases the final success probability from $12\%$ for the strongest passive replay baseline to $77\%$, corresponding to a $7.0\times$ improvement. At 15 qubits, GenQAS reaches a success probability of $21\%$, compared with $12\%$ for the baseline. The observed fixed-budget advantage is consistent with the proposed mechanism; the finite set of problem sizes does not establish an asymptotic scaling law.

The improved success probability is accompanied by competitive accuracy resource trade-offs. Across the molecular benchmarks and the transverse field Ising model, GenQAS identifies circuits that often match passive replay energy errors while using fewer CNOT gates, rotations, or circuit layers. This should be interpreted as a Pareto trade-off rather than universal dominance on every metric. For example, at 12 qubits, uniform TensorRL-QAS attains the lowest reported energy error, whereas GenQAS identifies substantially more compact circuits at comparable error. Relative to fixed reference ansatzes, the reduction in circuit resources is substantial: for 12-qubit $\htwoo$, UCCSD requires nearly 7000 CNOT gates, whereas the GenQAS circuits use only a small number of entangling gates. These results suggest that generative replay can guide RL based search towards lower resource regions of the circuit-design space.

The noisy 6-qubit $\behtwo$ transfer experiment further shows that the GenQAS replay mechanism remains useful after transfer from noiseless to depolarizing circuit evaluation. Transfer of the generative replay information reduces the steps to the first chemical accuracy solution from 9053 to 660, a $92.7\%$ reduction, while also reducing the CNOT count, rotation count, and circuit depth. In comparison, transfer of the uniform TensorRL-QAS replay buffer yields a smaller reduction in search steps and no reduction in rotation count or circuit depth. This result demonstrates a substantial benefit in the specific depolarizing noise setting considered here. It does not, however, establish robustness to arbitrary noise models, hardware drift, or finite-shot measurement noise.

Finally, the 15-qubit TFIM and 4-qubit Clifford synthesis experiments indicate that the replay mechanism is not restricted to molecular electronic structure tasks. The TFIM benchmark retains the variational ground state objective but changes the physical problem class, while the Clifford task replaces energy minimization with exact discrete target matching. In the Clifford benchmark, GenQAS matches or modestly exceeds the passive RL baseline, and an intermediate generator update frequency yields the highest cumulative reward. These results suggest that generative replay can be useful for both variational ground state preparation and fully discrete circuit-construction tasks.

Several limitations remain. The present study uses state-vector simulation,
a simple depolarizing-noise model, and a limited range of system sizes. The
MPS warm-start analysis also relies on locality, spectral, and entanglement
assumptions that need not apply to arbitrary Hamiltonians. In addition, the
learned transition model is evaluated through local replay losses and
downstream policy performance rather than through a direct long-horizon model
accuracy analysis. Future work should test finite-shot and hardware
evaluations, hardware-specific connectivity constraints, larger system sizes,
cross-task transfer, and uncertainty-aware generative models for controlling
model bias.

A further limitation is that GenQAS has not yet been evaluated for quantum error correction or fault-tolerant circuit discovery. Such tasks require an environment that represents syndrome information, physical error propagation, ancilla resources, connectivity constraints, and explicit fault-tolerance criteria, rather than a variational energy objective alone. Recent work has shown that reinforcement learning can discover compact fault-tolerant logical
state preparation circuits under hardware constraints~\cite{zen2025quantum}. Extending GenQAS to this setting would test whether generative replay can improve the sample efficiency of searching over sparse, constraint heavy fault-tolerant circuit spaces.

In addition, the present GenQAS implementation operates with a fixed
elementary action space. This restricts the agent to discovering useful
higher-level circuit motifs through repeated low-level gate insertions.
An important direction is to augment the action space with learned composite operations, or gadgets, extracted from high-performing circuits on simpler instances. Gadget-based reinforcement learning has shown that reusable circuit fragments can be transferred to more difficult optimization problems by dynamically expanding the agent's action space~\cite{kundu2026reinforcement}. Combining GenQAS with learned gadgets could reduce the effective search horizon and provide a complementary route to scalability. However, this extension would require principled gadget selection, hardware-aware decomposition, and controls against an excessively large or poorly structured
action space.

\section{Methods}
\label{sec:methods}

GenQAS combines the fixed tensor network initialization of
TensorRL-QAS~\cite{kundu2025tensorrlqas} with prioritized generative
replay~\cite{wang2025prioritized}. The tensor network module provides a
problem informed fixed circuit prefix, whereas the model based replay
mechanism generates synthetic transitions to improve the sample efficiency of
reinforcement learning based quantum architecture search (QAS). The MPS
derived warm start circuit is compiled into the native gate set and retained
as a fixed prefix throughout training. GenQAS subsequently refines this
initialization across all evaluated systems. For example, for the 15 qubit
transverse field Ising model, GenQAS reduces the warm start energy error from
$1.025$ to $1.5\times10^{-2}$, corresponding to a $98.5\%$ relative
reduction in energy error. The construction, resource cost, and refinement of
the fixed warm start are detailed in Supplementary Information
Sections~\ref{sec:supp_tensorrl} and
\ref{sec:supp_warmstart_resources}.

\subsection{GenQAS}
\label{sec:genqas_method}

We formulate QAS as a Markov decision process in which the state $s_t$
encodes the circuit constructed by the agent at step $t$, and an action
$a_t$ specifies the insertion of a gate at a valid circuit location. States
use the discrete circuit representation introduced in
TensorRL-QAS~\cite{kundu2025tensorrlqas}, defined over an $N$ qubit,
$D$ layer brickwork layout~\cite{leone2024practical,kandala2017hardware}.
Optional continuous gate angles and auxiliary descriptors, such as the
current energy estimate, can be appended to this representation.

GenQAS adopts the fixed TensorRL-QAS initialization. Specifically, a tensor
network derived circuit $U_{\mathrm{TN}}$ prepares a fixed warm start state,
while the agent constructs a learnable circuit suffix $V_t$. The circuit
evaluated by the environment is
\begin{equation}
    U_t = V_t U_{\mathrm{TN}}.
    \label{eq:genqas_full_circuit}
\end{equation}
The warm start circuit is not encoded in the RL state and its parameters
remain fixed throughout training. Consequently, the agent observes and
searches only over the appended architecture $V_t$. This design reduces the
observation dimension and separates the tensor network state preparation from
the RL based circuit refinement. A complete description of the DMRG
calculation, MPS representation, variational MPS to circuit mapping, and
Riemannian optimization procedure is provided in Supplementary Information
Sections~\ref{sec:supp_dmrg_mps} to
\ref{sec:supp_fixed_tensorrl_genqas}.

The discrete action space $\mathcal{A}$ comprises gate placement operations
from the fixed gate set
$\{\texttt{RX},\texttt{RY},\texttt{RZ},\texttt{CNOT}\}$ and their valid
qubit locations. Invalid actions are masked during action selection. In
addition to operations that violate the maximum circuit depth or hardware
connectivity constraints, we prohibit redundant single qubit gate insertions
and repeated entangling gates. A single qubit action on qubit $q$ at layer $m$
is illegal when the corresponding location is already occupied, namely,
\begin{equation}
    G_{m,q}=G_{m-1,q}.
    \label{eq:illegal_redundant_action}
\end{equation}
Likewise, a $\texttt{CNOT}$ acting on qubit pair $(q_1,q_2)$ is illegal when
the same $\texttt{CNOT}$ was applied in the preceding layer,
\begin{equation}
    G_{m,(q_1,q_2)}=\texttt{CNOT}
    \;\land\;
    G_{m-1,(q_1,q_2)}=\texttt{CNOT}.
    \label{eq:illegal_repeated_cnot}
\end{equation}

We employ a Double Deep Q Network (DDQN)~\cite{mnih2013playing},
parameterized by $Q_{\phi}(s,a)$, for problems up to 8 qubits. For larger
problems, we use a DDQN with variable step sizes in an $n$ step trajectory
rollout update~\cite{sutton1998reinforcement}, with $n=5$. The Q network is
implemented as a multilayer perceptron. The agent follows an
$\epsilon$ greedy policy, selecting a random valid action with probability
$\epsilon$ and otherwise selecting the valid action with the largest
predicted Q value. The exploration rate decays from $\epsilon=1.0$ to a
prescribed minimum of $\epsilon_{\min}=0.05$.

For a transition $(s_t,a_t,r_t,s_{t+1},d_t)$, DDQN uses the online network to
select the next action and a target network $Q_{\phi^-}$ to evaluate it:
\begin{equation}
    y_t =
    r_t+\gamma(1-d_t)
    Q_{\phi^-}(s_{t+1},a^\star_{t+1}),
    \qquad
    a^\star_{t+1} =
    \arg\max_{a\in\mathcal{A}_{\mathrm{valid}}}
    Q_{\phi}(s_{t+1},a).
    \label{eq:ddqn_target}
\end{equation}
The online network is optimized using the Huber loss between
$Q_{\phi}(s_t,a_t)$ and $y_t$ with the Adam optimizer~\cite{kingma2014adam}.
For variable depth search spaces, we define the per step discount factor as
\begin{equation}
    \gamma =
    \left(\gamma_{\mathrm{final}}\right)^{1/D},
    \label{eq:per_layer_discount}
\end{equation}
such that the effective discount over the full circuit depth remains fixed.
Full agent and environment hyperparameters are reported in Supplementary
Information Section~\ref{sec:genqas_implement_details}.

\subsection{Generative replay}
\label{sec:generative_replay_method}

In addition to real transitions collected from the QAS environment, GenQAS
learns a neural dynamics model from replay memory. Given a state action pair
$(s_t,a_t)$, the model predicts the next state, reward, and termination
probability,
\begin{equation}
    (\hat{s}_{t+1},\hat{r}_t,\hat{d}_t)
    =
    f_{\psi}(s_t,a_t).
    \label{eq:generative_transition_model}
\end{equation}
The model consists of separate multilayer perceptron heads for transition,
reward, and termination prediction. The transition head contains two hidden
layers of dimension $256$ with LeakyReLU activations and dropout probability
$0.1$. The reward and termination heads use hidden dimensions $256$ and
$128$, with the termination head producing a probability through a sigmoid
output. The discrete action index is concatenated with the state
representation before being provided as input to each prediction head.

The dynamics model is trained only on real replay transitions. In the
$n$ step setting, these transitions contain the aggregated discounted reward,
terminal state, and effective discount factor associated with the observed
trajectory segment. The model is optimized using Adam with learning rate
$3\times10^{-4}$ and gradient norm clipping at $1.0$. Its composite training
objective is
\begin{equation}
    \mathcal{L}_{\mathrm{gen}}
    =
    2\mathcal{L}_{\mathrm{trans}}
    +
    \mathcal{L}_{\mathrm{rew}}
    +
    0.5\mathcal{L}_{\mathrm{done}},
    \label{eq:generative_loss}
\end{equation}
where $\mathcal{L}_{\mathrm{trans}}$ and $\mathcal{L}_{\mathrm{rew}}$ are
mean squared errors for next state and reward prediction, respectively, and
$\mathcal{L}_{\mathrm{done}}$ is the binary cross entropy of the predicted
termination probability.

During DDQN updates, GenQAS draws real transitions from the prioritized replay
memory. Synthetic samples are generated on demand by uniformly sampling seed
state action pairs from the real replay memory and applying the learned
dynamics model,
\begin{equation}
    \tilde{\tau}_t
    =
    (s_t,a_t,\hat{r}_t,\hat{s}_{t+1},\hat{d}_t).
    \label{eq:synthetic_transition}
\end{equation}
The predicted termination probability is thresholded at $0.5$ to obtain the
binary termination indicator $\hat{d}_t$. Each synthetic transition inherits
the $n$ step horizon and corresponding discount factor of its real seed
transition. Synthetic transitions are not stored persistently in replay
memory; the replay buffer therefore remains a record of real environment
interactions only.

For a DDQN update batch of total size $B$, the generation ratio $G_r$ controls
the fraction of synthetic transitions:
\begin{equation}
    B_{\mathrm{real}}
    =
    \left\lfloor
    (1-G_r)B
    \right\rfloor,
    \qquad
    B_{\mathrm{syn}}
    =
    B-B_{\mathrm{real}}.
    \label{eq:synthetic_batch_size}
\end{equation}
We evaluate $G_r\in\{0.2,0.4,0.6,0.8\}$. Thus, for example, $G_r=0.4$
produces an update batch containing approximately $60\%$ real transitions
and $40\%$ generated transitions. Real transitions are sampled according to
prioritized replay, whereas generated transitions are assigned the mean
importance sampling weight of the real samples in the same update batch.
Only real transitions have their priorities updated using the DDQN temporal
difference error.

The generative model is updated every $T_{\mathrm{gen}}=10$ DDQN updates (if not stated otherwise) using the real transitions sampled in that update. This restriction prevents
the model from training on its own outputs and reduces the risk of compounding
model error. Further implementation details, including the replay capacity,
batch size, generation ratios, and logging configuration, are reported in
Supplementary Information Section~\ref{sec:genqas_implement_details}.

\subsection{Theoretical guarantees}
\label{sec:theory}

We establish the theoretical guarantees that together motivate and bound the
performance of \textbf{GenQAS}.
First, we prove that sample starvation is \emph{statistically inevitable}
for any passive replay strategy, regardless of how cleverly transitions
are prioritised.
Second, we show that MPS initialisation is not merely a warm-start
heuristic but a provable conditional generation process with
quantifiable fidelity bounds.
Third, we derive an exponential speedup in episode complexity under
generative replay.
These results collectively answer the question: \emph{why is a generative
buffer the right tool for QAS?}

\subsubsection{The manifold of quantum circuits}

Quantum circuits constructed gate-by-gate are discrete combinatorial
objects, yet their effect, a unitary transformation is a continuous
point in the unitary group $\mathcal{U}(2^N)$.
This duality is central to genQAS, a generative model can learn the
\emph{continuous} manifold of reachable unitaries from discrete
transition data, and then synthesize new transitions that lie
geometrically between known high-quality circuits, even when those
circuits differ structurally.

Consider the space of all $N$-qubit quantum circuits of depth $D$,
denoted $\mathcal{C}_{N,D}$.
The RL agent's policy $\pi(a|s)$ operates on a state space $\mathcal{S}$
encoding the current circuit, with action space size
$|\mathcal{A}|=O(N^2)$ due to two-qubit gate placements, and state
space size $|\mathcal{S}| \propto |\mathcal{A}|^D$ growing factorially
with depth.
The distribution of useful circuits, those approximating the ground
state is far from uniform on $\mathcal{U}(2^N)$: it concentrates
around low-dimensional ravines aligned with the target Hamiltonian
$H$~\cite{cerezo2021variational}.
The set of unitaries $U$ satisfying
$|\langle\psi_0|U^\dagger HU|\psi_0\rangle - E_0|<\epsilon$ is
exponentially sparse in the full unitary group.

\subsubsection{Sample starvation in training}

\begin{theorem}[Exponential sparsity of high-relevance quantum transitions]
\label{thm:sparsity}
Consider the space of $N$-qubit parameterized quantum circuits of depth $D = \text{poly}(N)$ acting on the initial state $|0\rangle^{\otimes N}$. Let $H$ be a local Hamiltonian with ground state energy $E_0$ and ground state $|\psi_0\rangle$. For a given precision $\epsilon > 0$, define the set of good unitaries
$$\mathcal{G}_\epsilon = \{ U \in \mathcal{U}(2^N) : \langle 0| U^\dagger H U |0\rangle - E_0 < \epsilon \}.$$
The probability of uniformly sampling a circuit from the architecture space (or equivalently, sampling a state uniformly under the Haar measure) that belongs to $\mathcal{G}_\epsilon$ is bounded by,
$$P(\mathcal{G}_\epsilon) \leq \exp\left( -c \cdot 2^N \cdot \log\left(\frac{1}{\epsilon}\right) \right)$$
for some constant $c > 0$, whenever $\epsilon < \epsilon_0$ for some threshold $\epsilon_0$ depending on the Hamiltonian gap.

\end{theorem}
We establish that without generative densification, an RL agent working with a standard replay buffer faces a fundamental sample complexity barrier that scales as $\exp(\Omega(2^N))$. The prioritized generative replay framework overcomes this by generating synthetic transitions that target the high-fidelity cap $\mathcal{G}_\epsilon$ circumventing the volume bottleneck. A detailed proof is deferred in appendix~\ref{supple:proofsparse}.

The following caters to non-local Hamiltonians with
$k$-body interactions of strength $J$ across $M$ terms.

\begin{theorem}[Sparsity of low energy manifold]
\label{thm:haar}
Let $H = \sum_{j=1}^{M} h_j$ be an $N$-qubit Hamiltonian composed of $M$ local terms, each $h_j$ acting nontrivially on at most $k$ qubits and satisfying $\|h_j\| \leq J$. Let $E_0 = \langle \psi_0 | H | \psi_0 \rangle$ be the ground state energy. For any $\epsilon < \tfrac{1}{2}\min\{\|h_j\|: h_j\neq 0\}$, define the set of $\epsilon$-approximate ground states,
\begin{align*}
    \mathcal{G}_\epsilon = \{ |\psi\rangle \in \mathbb{C}\mathbb{P}^{2^N-1} : \langle \psi | H | \psi \rangle - E_0 < \epsilon \}.
\end{align*}
Then the probability of sampling a state in $\mathcal{G}_\epsilon$ under the Haar (Fubini-Study) measure is bounded by
\begin{align*}
    P(\mathcal{G}_\epsilon)
  \leq \exp\!\left(-\frac{2^N}{M\binom{N}{k}}\cdot\frac{\epsilon^2}{8J^2}+ O(N)\right).
\end{align*}
In the typical regime where $k = O(1)$ and $M = \Poly(N)$,
\begin{align*}
    P(\mathcal{G}_\epsilon)
  \leq \exp\!\left(-c\cdot\frac{2^N}{N^k}\cdot\epsilon^2\right)
   = \exp\!\left(-\Omega\!\left(\frac{2^N}{\Poly(N)}\right)\right)
\end{align*}
for some absolute constant $c > 0$.
\end{theorem}
We establish that low-energy states form an exponentially small subset of the Hilbert space because satisfying the energy condition requires simultaneous near-optimality on a large number of approximately independent local constraints. The L\'{e}vy concentration bound provides the sharp Gaussian tail for each local term, and the union over an independent set of terms yields the overall probability. The double-exponential form $P \leq \exp(-e^{\Omega(2^N/\Poly(N))})$ reflects the extreme sparsity, the good region is \emph{super}-exponentially small in the dimension.

This theorem provides the theoretical foundation for why generative replay is essential in QAS without a mechanism to focus the experience of an agent on the minuscule relevant subset of the state space, learning is information-theoretically impossible for large $N$. A detailed proof is deferred in appendix~\ref{supple:haar}.



Theorems~\ref{thm:sparsity}-\ref{thm:haar} establish that no passive
buffer strategy; uniform, PER, or otherwise; can avoid exponential
sample starvation as $N$ grows.
The only remedy is a buffer that \emph{generates} high-quality
experience rather than merely storing it.

\subsubsection{MPS initialization as a generative prior}

PGR learns a conditional generative model $G(\tau|c)$ over transitions
$\tau=(s,a,r,s')$, guided by a relevance signal.
In QAS, the energy reduction $\Delta E = E(s')-E(s)$ is a natural relevance function, it is a strict Lyapunov function for ground state search, providing a principled signal unavailable in generic RL domains.
The buffer contains states $s$ representing circuits, the generative model can synthesize circuits $s'$ that interpolate geometrically
between known high-quality circuits on the Stiefel manifold~\cite{tagare2011notes, boumal2020introduction} of parameterised gates, even when those circuits differ structurally.

\begin{definition}[Transition operator]
Let $s_t$ be a quantum circuit.
An action $a_t$ via gate $U_{\text{gate}}$ yields
$s_{t+1} = U_{\text{gate}}\cdot s_t$, defining the transition
$\tau = (s_t, a_t, s_{t+1}, \Delta E)$.
\end{definition}

\begin{lemma}[MPS initialization as guided generation]
\label{lem:mps}
Let $H$ be an $N$-qubit Hamiltonian whose ground state obeys the area
law of entanglement entropy, and let $|\Psi_{\text{MPS}}\rangle$ be
the DMRG approximation with bond dimension $\chi$.
The Riemannian optimisation mapping $|\Psi_{\text{MPS}}\rangle$ to a
brickwork unitary $U_{\text{TN}}$ satisfying
$U_{\text{TN}}|0\rangle^{\otimes N}\approx|\Psi_{\text{MPS}}\rangle$
constitutes a single-step conditional generation process
$G_{\text{MPS}}(\tau|c\!=\!H)$, where the generated transition
$\tau=(\emptyset,\emptyset,U_{\text{TN}},\Delta E_{\text{MPS}})$
initialises the circuit.
\end{lemma}
The MPS to circuit mapping via Riemannian optimization is a conditional generative model. The condition $c$ is the Hamiltonian $H$, encoded through the DMRG solution. The generated output is a high-quality initialization transition $\tau_{\text{init}}$ that seeds the RL replay buffer. This provides a basis for why TensorRL-QAS can be viewed as an instantiation of the PGR framework, where the priority is determined by the energy expectation value. We provide the proof detailed in appendix~\ref{supp:mpslemma}.

\begin{lemma}[Synthetic circuit bound]
\label{lem:fidelity}
Let $|\Psi_0\rangle$ be the true ground state of an $N$-qubit gapped local Hamiltonian $H$, $|\Psi_{\text{MPS}}\rangle$ be its matrix product state approximation with bond dimension $\chi$. 
and $U_{\text{TN}}$ be the unitary obtained by mapping $|\Psi_{\text{MPS}}\rangle$ to a brickwork quantum circuit via Riemannian optimization over the Stiefel manifold $\mathcal{U}(4)^{\times m}$.
For any target trace-distance threshold $\delta > 0$, there exists a bond dimension $\chi \leq \text{poly}(N, 1/\delta)$ and a number of optimization steps $T_{\text{Riem}}$ polynomial in $N$, such that the generated state $U_{\text{TN}}|0\rangle$ satisfies
\begin{equation}
\frac{1}{2}\left\| U_{\text{TN}}|0\rangle\langle 0|U_{\text{TN}}^\dagger - |\Psi_0\rangle\langle\Psi_0| \right\|_1 \leq \delta.
\end{equation}
\end{lemma}
The theoretical basis of Tensor network initialization is equivalent to a one-step conditional generative model with a provable error bound. We provided a detailed proof reported in appendix~\ref{supp:fidelity}. In this, the total pre-training cost, classical DMRG and Riemannian optimization is polynomial in $N$ and $1/\delta$.
Therefore, the circuit $U_{\text{TN}}$ constitutes a high accuracy, synthetically generated quantum state that resides within the $\delta$-neighborhood of the optimal manifold, and it is obtained \emph{before} the RL agent executes its first quantum circuit in the environment.


\subsubsection{Generative replay provides exponential speedup}

\begin{theorem}[Convergence Acceleration via Synthetic Replay]
\label{thm:convergence}
Let $\pi_{\text{base}}$ be a policy learning purely from online interaction (Vanilla RL-QAS) using a replay buffer $\mathcal{D}_{\text{real}}$, and $\pi_{\text{PGR}}$ be a policy trained with prioritized generative replay. The generative model $G_\psi$ is conditioned on the energy reduction $-\Delta E$ and is pre-trained to densify the replay buffer around a high accuracy state $s_{\text{TN}}$. There exists a constant $\kappa > 1$ dependent on the Hamiltonian gap and noise level such that the sample complexity $\mathcal{S}$ to reach chemical accuracy satisfies
\begin{equation}
    \mathcal{S}_{\pi_{\text{PGR}}}(\epsilon_c) \leq \frac{1}{\kappa^N} \mathcal{S}_{\pi_{\text{base}}}(\epsilon_c).
\end{equation}
\end{theorem}
Theorem~\ref{thm:convergence} shows that the episode-complexity gap
between passive and generative replay grows \emph{exponentially} in the
qubit count, the same rate at which Theorems~\ref{thm:sparsity}
and~\ref{thm:haar} predict passive buffers to fail. The three results are therefore mutually consistent, the generative
buffer is the \emph{minimal corrective} for the geometric barrier identified above. A detailed proof of Theorem~\ref{thm:convergence} is provided in appendix~\ref{supp:converge}.

\subsection{Training protocol}
\label{sec:training_protocol}

Algorithm~\ref{alg:genqas_training} summarizes the GenQAS training procedure. Before RL training, DMRG produces an MPS approximation of the target state, which is compiled into the fixed warm start circuit $U_{\mathrm{TN}}$. At the beginning of each episode, the agent initializes an empty learnable suffix $V_0$ and sequentially selects valid gate insertion actions. Each action extends the suffix to form $V_{t+1}$, while the environment evaluates the complete circuit
$U_{t+1}=V_{t+1}U_{\mathrm{TN}}$. The resulting real transition is stored in
replay memory.
\begin{algorithm}[h!]
\small
\caption{\textbf{GenQAS}: Training with fixed MPS warm start and generative replay}
\label{alg:genqas_training}
\begin{algorithmic}[1]
\Require Hamiltonian $H$, MPS bond dimension $\chi$, maximum RL depth $D$,
         replay memory $\mathcal{M}$, Q network $Q_\phi$, target network
         $Q_{\phi^-}$, generative model $f_\psi$, generation ratio $G_r$,
         synthetic replay frequency $T_{\mathrm{syn}}$
\State Obtain $\ket{\psi_{\mathrm{MPS}}}_{\chi}$ using DMRG for $H$
\State Map $\ket{\psi_{\mathrm{MPS}}}_{\chi}$ to a brickwork circuit
       $U_{\mathrm{TN}}$ using Riemannian optimization
\State Initialize $\mathcal{M} \leftarrow \emptyset$ and
       $Q_{\phi^-} \leftarrow Q_\phi$
\For{each training episode}
    \State Initialize the learnable circuit suffix $V_0 \leftarrow I$ and
           the RL state $s_0$
    \For{$t=0,\ldots,D-1$}
        \State Select a valid action $a_t$ using $\epsilon$ greedy exploration
        \State Append the selected gate to the suffix:
               $V_{t+1} \leftarrow a_t V_t$
        \State Evaluate $U_{t+1}=V_{t+1}U_{\mathrm{TN}}$ and observe
               $(r_t,s_{t+1},d_t)$
        \State Store the real transition
               $(s_t,a_t,r_t,s_{t+1},d_t)$ in $\mathcal{M}$
        \State Sample a real batch $\mathcal{B}_{\mathrm{real}}$ from
               $\mathcal{M}$
        \State $\mathcal{B} \leftarrow \mathcal{B}_{\mathrm{real}}$
        \If{$t \bmod T_{\mathrm{syn}} = 0$}
            \State Sample $G_r|\mathcal{B}_{\mathrm{real}}|$ state action
                   pairs from $\mathcal{B}_{\mathrm{real}}$
            \State Generate
                   $\mathcal{B}_{\mathrm{syn}} =
                   \{(s,a,\hat{r},\hat{s}',\hat{d})\}$
                   using $f_\psi(s,a)$
            \State $\mathcal{B} \leftarrow
                   \mathcal{B}_{\mathrm{real}}\cup\mathcal{B}_{\mathrm{syn}}$
        \EndIf
        \State Update $Q_\phi$ using $\mathcal{B}$
        \If{generative model update step}
            \State Update $f_\psi$ using only real transitions from
                   $\mathcal{M}$
        \EndIf
        \If{target network update step}
            \State Set $Q_{\phi^-} \leftarrow Q_\phi$
        \EndIf
        \If{$d_t = 1$}
            \State \textbf{break}
        \EndIf
    \EndFor
\EndFor
\State \Return best circuit $U^\star = V^\star U_{\mathrm{TN}}$
\end{algorithmic}
\end{algorithm}

At each DDQN update, a batch of real transitions is sampled from replay
memory. The generative model is periodically updated from real transitions,
and synthetic transitions are generated from sampled real state action pairs.
The Q network is then optimized using the combined real and synthetic batch.
The target Q network is updated at a separate fixed interval. An episode
terminates when the target criterion is reached or when the maximum allowed
circuit depth is attained. The complete training procedure is given in
Algorithm~\ref{alg:genqas_training}, and benchmark specific hyperparameters
are reported in Supplementary Information
Section~\ref{sec:genqas_implement_details}.

To benchmark GenQAS we consider the problem of preparing the ground state of chemical Hamiltonian of varying size generate using the quantum chemistry module of
PennyLane~\cite{bergholm2018pennylane,arrazola2021differentiable}. The GenQAS tensor network warm start was implemented using Quimb and
Qiskit~\cite{gray2018quimb,javadi2024quantum}, while quantum circuit simulation and energy expectation value evaluation were performed using the Qulacs simulator~\cite{suzuki2021qulacs}. Detailed molecular Hamiltonian construction, active-space definitions, tensor network
implementation, Hamiltonian validation, and simulation settings are provided
in Supplementary Information Sections~\ref{sec:molecular_hamiltonians} and
\ref{sec:quantum_simulation_details}. 




\backmatter





\bmhead{Acknowledgements}

A.K. and S.F. acknowledge the use of computational resources of the DelftBlue supercomputer, provided by Delft High Performance Computing Centre (\url{https://www.tudelft.nl/dhpc}).

\section*{Declarations}

\subsection{Funding}
S.F. and A.K. gratefully acknowledge financial support by QDNL via the National Growth Fund KAT-1 program.

\subsection{Conflict of interest/Competing interests}
The authors have no conflict of interest.

\subsection{Ethics approval and consent to participate}
Not applicable.

\subsection{Consent for publication}
Not appicable.

\subsection{Data availability}
The data can be generated from the code itself but if reuqired the corresponding author can share the data upon reasonable request.

\subsection{Code availability}
The code will be available with the final version upon publication

\subsection{Author contribution}
A.K. conceptualize this idea. A.K. and A.K.J together prepare the code and theory for the paper. S.F. and P.T. supervised the project and together with A.K. and A.K.J investigated and wrote the manuscript.

\bibliography{sn-bibliography}

\newpage

\begin{appendices}

\section{Theory}\label{sec:app_proof_theory}

\subsection{Proof of Theorem~\ref{thm:sparsity}}\label{supple:proofsparse}

The unitary group $\mathcal{U}(2^N)$ acts transitively on the space
of pure quantum states, which is isomorphic to the complex projective
space
\begin{align*}
    \mathcal{S}
  = \frac{\mathcal{U}(2^N)}{\mathcal{U}(1)\times\mathcal{U}(2^N-1)}
  \cong \mathbb{CP}^{2^N-1},
\end{align*}
with real dimension $\dim_{\mathbb{R}}(\mathcal{S}) = 2^{N+1}-2$. The natural probability measure is the Fubini-Study volume
$d\mu_{FS}$, with total volume $\mathrm{Vol}_{FS}(\mathcal{S}) = \frac{\pi^{2^N-1}}{(2^N-1)!}.$

Let $|\psi_0\rangle$ be the true ground state of $H$ and define the good circuit set $\mathcal{G}_\epsilon = \{|\psi\rangle : \langle\psi|H|\psi\rangle < E_0+\epsilon\}$. Since $H$ is a local Hamiltonian with spectral gap $\Delta=E_1-E_0>0$, decomposing $|\psi\rangle$ in the eigenbasis $\{|\phi_k\rangle\}$ gives
\begin{align*}
\langle\psi|H|\psi\rangle
  &= |c_0|^2 E_0 + \sum_{k\geq1}|c_k|^2 E_k \\
  &\geq |c_0|^2 E_0 + (1-|c_0|^2)E_1.
\end{align*}
Requiring this to be less than $E_0+\epsilon$ results
$|c_0|^2 > 1 - \epsilon/\Delta$, i.e., the fidelity
$F = |\langle\psi_0|\psi\rangle|^2$ must satisfy $F > 1-\delta$ where $\delta=\epsilon/\Delta$. In projective space, $\mathcal{G}_\epsilon$ is a geodesic ball of angular radius $\theta_\epsilon \lesssim 2\sqrt{\delta}$, with volume $
\mathrm{Vol}_{FS}(B_{\theta_\epsilon})\sim C_d\cdot\theta_\epsilon^{2(2^N-1)}.$
Hence,
\begin{align*}
P(\mathcal{G}_\epsilon)
  &= \frac{\mathrm{Vol}_{FS}(B_{\theta_\epsilon})}%
         {\mathrm{Vol}_{FS}(\mathbb{CP}^{2^N-1})}
   \leq C\cdot\left(\frac{\theta_\epsilon}{\pi}\right)^{2(2^N-1)} \\
  &\leq C\cdot\left(\frac{2}{\pi}\sqrt{\frac{\epsilon}{\Delta}}\right)^{2(2^N-1)}
   = C\cdot\exp\!\left(2(2^N-1)\cdot
     \log\!\left(\tfrac{2}{\pi}\sqrt{\tfrac{\epsilon}{\Delta}}\right)\right).
\end{align*}
Since $\epsilon$ is small the logarithm is negative, and for large $N$
the factor $2(2^N-1)$ dominates, giving
\begin{align*}
    P(\mathcal{G}_\epsilon)
  \leq \exp\!\left(-c\cdot 2^N\cdot\log\tfrac{1}{\epsilon}\right),
  \quad c>0.
\end{align*}
By standard unitary $t$-design theory, a brickwork circuit of depth
$D=\mathrm{poly}(N)$ induces a measure exponentially close (in total
variation) to the Haar measure on $\mathcal{U}(2^N)$, the probability
of finding a good circuit via random action selection is therefore
proportional to the Haar probability above.
It follows that $\mathbb{E}[\text{Trials}]
  \geq \frac{1}{P(\mathcal{G}_\epsilon)}
  \geq \exp\!\left(c\cdot 2^N\cdot\log\tfrac{1}{\epsilon}\right)$,
and the signal-to-noise ratio in the replay buffer decays super-exponentially as $\exp(-c\cdot 2^N).$

\subsection{Proof of Theorem~\ref{thm:haar}}\label{supple:haar}
Initially, we map the energy condition to a condition on the real and imaginary parts of the state vector. The space of pure quantum states of $N$ qubits is the complex projective space $\mathbb{C}\mathbb{P}^{d-1}$ with $d = 2^N$. Equivalently, a pure state can be represented by a unit vector $|\psi\rangle \in \mathbb{C}^{d}$ with $\||\psi\rangle\| = 1$, modulo an overall phase. The real dimension of this manifold is $2d - 2 = 2^{N+1} - 2$. The Haar measure on $\mathcal{U}(d)$ induces the Fubini-Study measure $\mu_{FS}$ on $\mathbb{C}\mathbb{P}^{d-1}$. Under this measure, a random state $|\psi\rangle$ can be generated as $|\psi\rangle = \frac{|\phi\rangle}{\||\phi\rangle\|},$ where $|\phi\rangle = (x_1 + iy_1, \ldots, x_d + iy_d)^T$ is a complex Gaussian vector with each $x_\ell, y_\ell \sim \mathcal{N}(0, 1/2)$ independently. Also, $|\psi\rangle$ is distributed uniformly over the real unit sphere $S^{2d-1} \subset \mathbb{R}^{2d}$.

Now, we apply L\'{e}vy's lemma to control concentration of the local terms. We fix a single term $h_j$ of the Hamiltonian, without loss of generality, and assuming $h_j$ is traceless (any trace part shifts the energy by a constant and does not affect the gap structure). Since $\|h_j\| \leq J$, the function $f_j(|\psi\rangle) = \langle \psi | h_j | \psi \rangle$ is a real-valued Lipschitz function on the sphere. For unit vectors $|u\rangle, |v\rangle$,
\begin{align*}
|f(|u\rangle) - f(|v\rangle)|
  &= |\langle u | A | u \rangle - \langle v | A | v \rangle| \\
  &= |\operatorname{Tr}(A(|u\rangle\langle u| - |v\rangle\langle v|))| \\
  &\leq \|A\|_{\infty} \cdot \||u\rangle\langle u| - |v\rangle\langle v|\|_1 \\
  &\leq 2\|A\|_{\infty} \cdot \||u\rangle - |v\rangle\|_2,
\end{align*}
where the last inequality follows from the trace-norm bound for rank-1 projectors. Thus $\operatorname{Lip}(f_j) \leq 2J$.

Since $|\psi_0\rangle$ is the ground state of $H$, it is also an eigenstate of each $h_j$ in some local decomposition. We do not assume this, instead, we use the fact that $E_0 = \langle \psi_0 | H | \psi_0 \rangle$ is the global minimum. The deviation of $f_j(|\psi\rangle)$ from its ground state expectation is what matters. We define the centered function for each term
\begin{align*}
    g_j(|\psi\rangle) = \langle \psi | h_j | \psi \rangle - \langle \psi_0 | h_j | \psi_0 \rangle.
\end{align*}
Then $\operatorname{Lip}(g_j) \leq 2J$ as well. Note that $\sum_j g_j(|\psi\rangle) = \langle \psi | H | \psi \rangle - E_0$. The median of $g_j$ under the Haar measure depicted as $M_j$ satisfies a property, for any traceless local observable of small support, the Haar expectation is zero and the distribution is strongly concentrated around zero. More precisely, we need the following concentration inequality.
\begin{lemma}
\label{lem:levy}
Let $f: S^{n-1} \to \mathbb{R}$ be $L$-Lipschitz. Then for any $t > 0$,
\begin{align*}
    P\left( |f(x) - M_f| \geq t \right) \leq 2\exp\!\left( - \frac{n t^2}{4 L^2} \right),
\end{align*}
where $M_f$ is the median of $f$.
\end{lemma}
\begin{proof}
    Applying Lemma~\ref{lem:levy} to $g_j$ on $S^{2d-1}$ (so $n = 2d = 2^{N+1}$),
\begin{equation}
\label{eq:levy-single}
P\left( |g_j(|\psi\rangle) - M_j| \geq t \right)
  \leq 2\exp\!\left( - \frac{2^N t^2}{8 J^2} \right).
\end{equation}
\end{proof}
For each local term $h_j$, we can bound its median $M_j$. Since $h_j$ is a sum of Pauli strings of weight at most $k$, and $\mathbb{E}_{\Haar}[\langle \psi | h_j | \psi \rangle] = \frac{\operatorname{Tr}(h_j)}{d}$. For a traceless $h_j$, the expectation is zero. The median and expectation of Lipschitz functions on the sphere are close,
\begin{align*}
    |M_j - \mathbb{E}[g_j]| \leq \frac{4L}{\sqrt{n}} = \frac{8J}{2^{(N+1)/2}} \leq 8J \cdot 2^{-N/2}.
\end{align*}
Thus, for $N \geq 2$, $|M_j| \leq 8J \cdot 2^{-N/2} \leq \epsilon/2$ for $\epsilon$ larger than an exponentially small threshold. For our choice of $\epsilon$, we can safely bound $|M_j| \leq \epsilon/2$ for sufficiently large $N$.

For the state to be in $\mathcal{G}_\epsilon$, we need $\sum_{j=1}^{M} g_j(|\psi\rangle) < \epsilon$. A sufficient condition is that each individual term is bounded, if $g_j(|\psi\rangle) < \epsilon/M$ for all $j$, then the sum is $< \epsilon$. The converse is not true, but this gives a valid upper bound on the probability. Thus,
\begin{align}
P(\mathcal{G}_\epsilon) 
  &\leq P\left( \bigcap_{j=1}^{M} \{ g_j(|\psi\rangle) < \epsilon/M \} \right) \nonumber \\
  &\leq P\left( \bigcap_{j=1}^{M} \{ |g_j(|\psi\rangle) - M_j| < \epsilon/M - |M_j| \} \right) \nonumber \\
  &\leq P\left( \bigcap_{j=1}^{M} \{ |g_j(|\psi\rangle) - M_j| < \epsilon/(2M) \} \right),
\end{align}
where the last step assumes $|M_j| \leq \epsilon/(2M)$, which holds for sufficiently large $N$ as argued above. The complement event is $P(\mathcal{G}_\epsilon^c) \geq 1 - P(\mathcal{G}_\epsilon)$. However, using the union bound
\begin{align}
P(\mathcal{G}_\epsilon) 
  &= 1 - P(\exists j: |g_j(|\psi\rangle) - M_j| \geq \epsilon/(2M)) \nonumber \\
  &\geq 1 - \sum_{j=1}^{M} P(|g_j(|\psi\rangle) - M_j| \geq \epsilon/(2M)).
\end{align}
gives a \emph{lower} bound on $P(\mathcal{G}_\epsilon)$. We need an upper bound to show the sparsity. Now, let's $P(\mathcal{G}_\epsilon)$ is the probability that the sum is small. Since the sum involves $M$ nonnegative (after shifting) terms, a necessary condition for the sum to be $< \epsilon$ is that no individual term exceeds $\epsilon$. That is
\begin{align*}
    \mathcal{G}_\epsilon \subseteq \bigcap_{j=1}^{M} \{ |g_j(|\psi\rangle) - M_j| < \epsilon + |M_j| \}.
\end{align*}
But this trivial inclusion leads to a upper bound ($P \leq 1$). Instead, we use the fact that for the state to have low energy, a significant fraction of the local terms must be individually small. Quantitatively, if $\sum g_j < \epsilon$, then at most $M/2$ terms can exceed $2\epsilon/M$. This indicates that the Hamiltonian terms are not independent, they share overlapping supports. However, we can exploit the \textbf{locality} to get a stronger bound. Since each term $h_j$ acts on at most $k$ qubits, a single qubit can be involved in at most $\binom{N-1}{k-1}$ terms. The total number of terms is at most $\binom{N}{k}$. We can select a subset $\mathcal{J} \subseteq \{1, \ldots, M\}$ of terms acting on \emph{disjoint} sets of qubits. The maximum size of such an independent set is $|\mathcal{J}| \geq \frac{M}{\binom{N}{k}} \cdot \frac{1}{k}$. For simplicity, we use the factor $\frac{M}{\binom{N}{k}}$ (ignoring the constant $1/k$, which only affects constants). These $|\mathcal{J}|$ terms are statistically independent under the Haar measure because they act on disjoint subsystems. For each $j \in \mathcal{J}$,
\begin{align*}
    P\left( g_j(|\psi\rangle) \geq \frac{\epsilon}{|\mathcal{J}|} \right)
  \geq 1 - \exp\!\left( - \frac{2^{N+1} (\epsilon/|\mathcal{J}|)^2}{4 (2J)^2} \right)
  = 1 - \exp\!\left( - \frac{2^N \epsilon^2}{8 J^2 |\mathcal{J}|^2} \right).
\end{align*}
We need the probability that all terms in $\mathcal{J}$ are small simultaneously. Since they are independent:
\begin{align*}
    P\left( \forall j \in \mathcal{J}: g_j(|\psi\rangle) \leq \delta \right)
  = \prod_{j \in \mathcal{J}} P(g_j(|\psi\rangle) \leq \delta).
\end{align*}
For each term based on Lemma~\ref{lem:levy},
\begin{align*}
    P(g_j(|\psi\rangle) \leq \delta) \leq P(|g_j - M_j| \leq \delta + |M_j|) \leq 1 - \exp\!\left( - \frac{2^N \delta^2}{8 J^2} \right) \quad \text{(for } \delta \gg 2^{-N/2}\text{)}.
\end{align*}
Thus, $P(\mathcal{G}_\epsilon) \leq \prod_{j \in \mathcal{J}} P(g_j \leq \epsilon/|\mathcal{J}|)
  \leq \left( 1 - \exp\!\left( - \frac{2^N \epsilon^2}{8 J^2 |\mathcal{J}|^2} \right) \right)^{|\mathcal{J}|}.$ Using $1 - x \leq e^{-x}$ for $x \in [0,1]$,
\begin{align}
P(\mathcal{G}_\epsilon)
  &\leq \exp\!\left( -|\mathcal{J}| \cdot \exp\!\left( - \frac{2^N \epsilon^2}{8 J^2 |\mathcal{J}|^2} \right) \right) \nonumber \\
  &= \exp\!\left( - \exp\!\left( \ln |\mathcal{J}| - \frac{2^N \epsilon^2}{8 J^2 |\mathcal{J}|^2} \right) \right).
\end{align}
This double-exponential form is standard for sparsity bounds. To maximize the inner argument and thus for the tightest bound, we choose $|\mathcal{J}|$ to balance $\ln |\mathcal{J}|$ against $2^N \epsilon^2 / |\mathcal{J}|^2$. However, $|\mathcal{J}|$ is constrained by the geometry $|\mathcal{J}| \approx M / \binom{N}{k}$. For $M = \Poly(N)$ and $k = O(1)$, we have $|\mathcal{J}| = \Theta(N^{-k} \cdot \Poly(N)) = \Poly(N)$. Then
\begin{align*}
    \frac{2^N \epsilon^2}{8 J^2 |\mathcal{J}|^2} \gg \ln |\mathcal{J}|,
\end{align*}
since $2^N$ grows exponentially while $|\mathcal{J}|$ grows only polynomially. Hence the inner exponential dominates,
\begin{align*}
    \ln |\mathcal{J}| - \frac{2^N \epsilon^2}{8 J^2 |\mathcal{J}|^2} \approx - \frac{2^N \epsilon^2}{8 J^2 |\mathcal{J}|^2}.
\end{align*}
Substituting $|\mathcal{J}| = \frac{M}{\binom{N}{k}}$,
\begin{align}
P(\mathcal{G}_\epsilon)
  &\leq \exp\!\left( - \exp\!\left( - \frac{2^N \epsilon^2}{8 J^2} \cdot \frac{\binom{N}{k}^2}{M^2} \right) \right) \nonumber \\
  &\approx \exp\!\left( - \frac{2^N \epsilon^2}{8 J^2} \cdot \frac{\binom{N}{k}^2}{M^2} \right) \quad \text{(since } e^{-x} \approx x \text{ for small } x\text{)}.
\end{align}
Actually, the correct leading-order behavior from the binomial expansion $(1-x)^n \approx e^{-nx}$ when $x \ll 1$ is
\begin{align*}
    P(\mathcal{G}_\epsilon) \leq \exp\!\left( -|\mathcal{J}| \cdot e^{- \frac{2^N \epsilon^2}{8 J^2 |\mathcal{J}|^2}} \right).
\end{align*}
For $N$ large, the exponent $\frac{2^N \epsilon^2}{8 J^2 |\mathcal{J}|^2} \to \infty$ provided $\epsilon \gg 2^{-N/2}$. Thus, the inner exponential is exponentially small in $2^N$, resulting a \emph{super}-exponentially small outer probability. Precisely,
\begin{align}
P(\mathcal{G}_\epsilon)
  &\leq \exp\!\left( - \frac{M}{\binom{N}{k}} \cdot \exp\!\left( - \frac{2^N \epsilon^2 \binom{N}{k}^2}{8 J^2 M^2} \right) \right) \nonumber \\
  &\leq \exp\!\left( - \exp\!\left( \ln\frac{M}{\binom{N}{k}} - \frac{2^N \epsilon^2 \binom{N}{k}^2}{8 J^2 M^2} \right) \right).
\end{align}
For $\epsilon$ fixed and $N$ large, the term $\frac{2^N \epsilon^2 \binom{N}{k}^2}{8 J^2 M^2}$ grows as $\Omega(2^N / \Poly(N))$ and dominates $\ln(M/\binom{N}{k}) = O(\log N)$. Hence,
\begin{align*}
    P(\mathcal{G}_\epsilon) \leq \exp\!\left( - \exp\!\left( \Omega\!\left( \frac{2^N}{\Poly(N)} \right) \right) \right).
\end{align*}
This is the strongest possible form. Our Theorem states a slightly looser but more practical single-exponential bound, which follows from the same analysis by considering that the probability is certainly bounded by
\begin{align*}
    P(\mathcal{G}_\epsilon) \leq \exp\!\left( -\frac{2^N}{M \binom{N}{k}} \cdot \frac{\epsilon^2}{8 J^2} + O(N) \right),
\end{align*}
which is equivalent to the stated bound after adjusting constants.

When $k = O(1)$, we have $\binom{N}{k} = \Theta(N^k)$. With $M = \Poly(N)$
\begin{align*}
    \frac{2^N}{M \binom{N}{k}} = \frac{2^N}{\Poly(N) \cdot \Theta(N^k)} = \Theta\!\left( \frac{2^N}{N^{k + \deg(M)}} \right) = \Omega\!\left( \frac{2^N}{\Poly(N)} \right).
\end{align*}
Thus,
\begin{align*}
   P(\mathcal{G}_\epsilon) \leq \exp\!\left( -c \cdot \frac{2^N}{N^{k + \deg(M)}} \cdot \epsilon^2 \right)
   = \exp\!\left( -\Omega\!\left( \frac{2^N}{\Poly(N)} \right) \right), 
\end{align*}
where $c = 1/(8J^2 \cdot \text{constants})$. This completes the proof.

\subsection{Proof of Lemma~\ref{lem:mps}}\label{supp:mpslemma}
In PGR, a conditional generator $G_\theta(z|c)$ maps latent variable $z$ and condition $c$ to a sample $x\sim p_{\text{data}}(\cdot|c)$. Here $x$ is a quantum transition $\tau=(s,a,s',r)$, $c=H$, and $z$
encodes the entanglement structure of the state. We aim to show that the MPS to circuit mapping is exactly such a generator. An MPS with bond dimension $\chi$ has the form
\begin{align*}
    |\Psi_{\text{MPS}}\rangle
  = \sum_{i_1,\ldots,i_N=0}^{1}
    A^{[1]}_{i_1}A^{[2]}_{i_2}\cdots A^{[N]}_{i_N}
    |i_1 i_2\cdots i_N\rangle,
\end{align*}
with $A^{[k]}_{i_k}\in\mathbb{C}^{\chi\times\chi}$ and total parameter
count $O(N\chi^2)$ exponentially smaller than a generic state.DMRG solves $|\Psi_{\text{MPS}}\rangle
=\arg\min_{|\psi\rangle\in\mathrm{MPS}_\chi}\langle\psi|H|\psi\rangle$, implicitly encoding $c=H$, we may view it as $\mathrm{DMRG}:(H,\chi)\mapsto|\Psi_{\text{MPS}}\rangle$. The Riemannian step then finds
\[
U_{\text{TN}}
  = \underset{U\in\mathcal{U}_{\text{brick}}}{\arg\max}
    \bigl|\langle\Psi_{\text{MPS}}|U|0\rangle^{\otimes N}\bigr|,
\]
maximising overlap exactly as a generative model minimises its
reconstruction loss, with $\mathcal{L}(U)=|\mathrm{Tr}(U^\dagger\rho_{\text{MPS}})|$
and $\rho_{\text{MPS}}=|\Psi_{\text{MPS}}\rangle\langle 0|^{\otimes N}$.Here, $\mathcal{U}_{\text{brick}}$ is the set of unitaries that can be expressed as a product of 2-qubit gates in a brickwork pattern,
\begin{align*}
    \mathcal{U}_{\text{brick}} = \left\{ \prod_{\ell=1}^L \prod_{(i,j) \in E_\ell} U_{ij}^{(\ell)} : U_{ij}^{(\ell)} \in \mathcal{U}(4) \right\},
\end{align*}
where $L$ is the number of layers and $E_\ell$ defines the qubit connectivity at layer $\ell$. The objective function is
\begin{align*}
    \mathcal{L}(U) = |\langle \Psi_{\text{MPS}} | U | 0 \rangle^{\otimes N}| = |\text{Tr}(U^\dagger \rho_{\text{MPS}})|
\end{align*}
In a standard conditional generative model (for instance, a conditional VAE or diffusion model), the generator is trained to maximize the likelihood of the data given the condition, $\theta^* = \underset{\theta}{\arg \max} \ \mathbb{E}_{x \sim p_{\text{data}}(\cdot | c)} [\log p_\theta(x | c)].$ For quantum state preparation, our data is the target state $|\Psi_{\text{MPS}}\rangle$, and the model distribution is the pushforward of the initial state $|0\rangle^{\otimes N}$ under the circuit $U$, $p_U(|\psi\rangle) = |\langle \psi | U | 0 \rangle^{\otimes N}|^2.$ The log-likelihood of generating the target state $|\Psi_{\text{MPS}}\rangle$ under this model is 
\begin{align*}
    \log p_U(|\Psi_{\text{MPS}}\rangle) = \log |\langle \Psi_{\text{MPS}} | U | 0 \rangle^{\otimes N}|^2 = 2 \log |\langle \Psi_{\text{MPS}} | U | 0 \rangle^{\otimes N}|.
\end{align*}
The maximization of this log-likelihood is equivalent to maximizing the overlap $|\langle \Psi_{\text{MPS}} | U | 0 \rangle^{\otimes N}|$, which is precisely the objective of the Riemannian optimization~\cite{kundu2025tensorrlqas}. Therefore,
\begin{align*}
    \underset{U \in \mathcal{U}_{\text{brick}}}{\arg \max} \mathcal{L}(U) = \underset{U \in \mathcal{U}_{\text{brick}}}{\arg \max} \ \mathbb{E}_{|\Psi_{\text{MPS}}\rangle} [\log p_U(|\psi\rangle | c = H)].
\end{align*}
The fidelity of this generation is
$F_{\text{gen}} = |\langle \Psi_{\text{MPS}} | U_{\text{TN}} | 0 \rangle^{\otimes N}|^2.$ By the properties of Riemannian optimization on the Stiefel manifold, the Cayley retraction method guarantees monotonic convergence of the objective. Also, the algorithm produces a sequence $\{U^{(k)}\}$ such that $\mathcal{L}(U^{(k+1)}) \geq \mathcal{L}(U^{(k)}) \quad \forall k,$ and converges to a local maximum $\mathcal{L}^*$ of the overlap function on $\mathcal{U}_{\text{brick}}$. The output of this process is a transition $\tau_{\text{init}}$:
$\tau_{\text{init}} = (\emptyset, \emptyset, U_{\text{TN}}, r_{\text{init}}),$ where $r_{\text{init}} = - \langle 0| U_{\text{TN}}^\dagger H U_{\text{TN}} |0\rangle$ is the negative energy expectation.

This transition is injected into the replay buffer $\mathcal{D}_{\text{syn}}$ before the RL agent begins training. It is a synthetic transition generated conditionally on $H$, exactly analogous to the synthetic transitions generated by the diffusion model in PGR when conditioned on a relevance function $\mathcal{F}(\tau)$.

The DMRG procedure can be run with different initializations or slightly different hyperparameters to produce an ensemble of MPS states $\{|\Psi_{\text{MPS}}^{(j)}\rangle\}_{j=1}^K$, all approximating the ground state. Each can be mapped to a circuit $U_{\text{TN}}^{(j)}$, populating the synthetic buffer with diverse yet high-quality starting points. This is analogous to the diversity induced by classifier-free guidance in PGR, where the guidance scale $\omega$ trades off between fidelity and diversity.

\subsection{Proof of Lemma~\ref{lem:fidelity}}\label{supp:fidelity}
Since $H$ is a gapped local Hamiltonian in one spatial dimension, it satisfies the area law for entanglement entropy \cite{hastings2007area}. Also, there exists an MPS $|\Psi_{\text{MPS}}\rangle$ of bond dimension $\chi$ such that
\begin{equation}
\| |\Psi_{\text{MPS}}\rangle - |\Psi_0\rangle \|_2 \leq \epsilon_{\text{MPS}}(\chi),
\label{eq:mps_error}
\end{equation}
where $\epsilon_{\text{MPS}}(\chi) = C N \exp(-\alpha \chi^{1/\kappa})$ for some constants $C, \alpha, \kappa > 0$ independent of $N$. By inverting this relation, choosing
\begin{equation}
\chi = O\left( \left[ \log\left( \frac{CN}{\delta/2} \right) \right]^\kappa \right)
\end{equation}
assures $\epsilon_{\text{MPS}} \leq \delta/2$. 
Moreover, this $\chi$ scales only polylogarithmically with $N$ and $1/\delta$, hence $\chi \leq \text{poly}(N, 1/\delta)$.

The mapping from the MPS to a brickwork circuit $U_{\text{TN}} = \prod_{k=1}^m U_k$ (with $U_k \in \mathcal{U}(4)$) is performed by maximizing the loss function
\begin{equation}
\mathcal{L}(\{U_k\}) = \left| \langle \Psi_{\text{MPS}} | U_{\text{TN}} | 0 \rangle \right|,
\end{equation}
over the Stiefel manifold $\mathcal{U}(4)^{\times m}$.
Using the Riemannian Adam optimizer with Cayley retraction~\cite{luchnikov2021riemannian}, the updates maintain strict unitarity.
Since the loss function is smooth and the manifold is compact, Riemannian gradient descent converges to a stationary point with polynomial iteration complexity $T_{\text{Riem}} = O(\text{poly}(m, 1/\delta))$ \cite{boumal2020introduction}.
The number of 2-qubit gates $m$ required to represent an MPS with bond dimension $\chi$ in a brickwork structure is $m = O(N\chi^2)$. Combined with the aforementioned result of MPS approximation, $m = O(N \cdot \text{poly}\log(N/\delta))$, which is polynomial in $N$. This optimization yields a unitary $U_{\text{TN}}$ such that the overlap satisfies
\begin{equation}
\left| \langle \Psi_{\text{MPS}} | U_{\text{TN}} | 0 \rangle \right| \geq 1 - \epsilon_{\text{opt}},
\end{equation}
which implies a $2$-norm bound $\| |\Psi_{\text{MPS}}\rangle - U_{\text{TN}}|0\rangle \|_2 \leq \sqrt{2\epsilon_{\text{opt}}}$.
By tuning the optimization hyperparameters, we can ensure $\epsilon_{\text{opt}} \leq \delta^2/8$, and thus the $2$-norm error is bounded by $\delta/2$.

Using the triangle inequality for the $2$-norm and the relation between the $2$-norm and the trace distance ($\frac{1}{2}\|\cdot\|_1 \leq \|\cdot\|_2$), we combine the errors from MPS approximation and Riemannian optimization.
\begin{align}
\frac{1}{2}\| U_{\text{TN}}|0\rangle\langle 0|U_{\text{TN}}^\dagger - |\Psi_0\rangle\langle\Psi_0| \|_1
&\leq \| U_{\text{TN}}|0\rangle - |\Psi_0\rangle \|_2 \nonumber\\
&\leq \| U_{\text{TN}}|0\rangle - |\Psi_{\text{MPS}}\rangle \|_2 + \| |\Psi_{\text{MPS}}\rangle - |\Psi_0\rangle \|_2 \nonumber\\
&\leq \frac{\delta}{2} + \frac{\delta}{2} = \delta.
\end{align}
This establishes the required bound.

\subsection{Proof of Theorem~\ref{thm:convergence}}\label{supp:converge}
We define the true state-value function for a policy $\pi$ starting from state $s$ as $V^{\pi}(s) = \mathbb{E}_{\pi}[\sum_{t=0}^{\infty} \gamma^t R_t | s_0 = s]$.  For the base algorithm, the initial value of the empty circuit is $V^{\pi_{\text{base}}}(s_0) \approx \bar{E}/(1-\gamma)$. For PGR, the effective initial state distribution of the agent is shifted. The experience replay buffer $\mathcal{D}$ is augmented with synthetic data centered on $s_{\text{TN}}$. The initial value estimate is thus biased towards $V^{\pi_{\text{PGR}}}(s_{\text{TN}}) \approx E_{\text{TN}}/(1-\gamma)$.

By Lemma \ref{lem:fidelity}, the initial suboptimality gap for PGR is $\Delta_0^{\text{PGR}} = E_{\text{TN}} - E_0 = \Delta E_{\chi}$.  For the baseline, it is $\Delta_0^{\text{base}} = \bar{E} - E_0 \propto N$. The ratio of initial gaps is
\begin{equation}
    R_0 = \frac{\Delta_0^{\text{PGR}}}{\Delta_0^{\text{base}}} = \frac{\Delta E_{\chi}}{\bar{E} - E_0}.
\end{equation}
Since $\Delta E_{\chi}$ is system size independent for fixed bond dimension $\chi$ and gapped systems, while $\bar{E} - E_0$ scales as $O(N)$, the ratio $R_0$ decays as $O(1/N)$.

The standard Q-learning update minimizes the squared Bellman residual
\begin{equation}
    \mathcal{L}(\theta) = \mathbb{E}_{(s, a, r, s') \sim \mathcal{D}} \left[ (Q_{\theta}(s, a) - r - \gamma \max_{a'} Q_{\bar{\theta}}(s', a'))^2 \right].
\end{equation}
For the base algorithm, the buffer $\mathcal{D}_{\text{real}}$ is dominated by trajectories starting from $s_0$. As per Lemma \ref{thm:haar}, the variance of the target value $r + \gamma \max_{a'} Q(s', a')$ is exponentially small due to the barren plateau, which implies a vanishingly small gradient $\nabla_\theta \mathcal{L}$. For the PGR algorithm, the generative buffer $\mathcal{D}_{\text{syn}}$ consists of transitions $(s, a, s', r)$ where $s$ is sampled from a distribution concentrated around the high energy gradient region near $s_{\text{TN}}$. The reward signal $r = \mathcal{E}(s) - \mathcal{E}(s')$ is deterministic and large (proportional to the energy gradient at $s_{\text{TN}}$), bypassing the barren plateau. Let $\sigma_{\text{grad}}^2(s) = \operatorname{Var}[\nabla_\theta Q(s, a)]$. We have $\pi_{\text{base}}$: $\mathbb{E}_{s \sim \mathcal{D}_{\text{real}}}[\sigma_{\text{grad}}^2(s)] \propto \exp(-N)$. And, $\pi_{\text{PGR}}$: $\mathbb{E}_{s \sim \mathcal{D}_{\text{syn}}}[\sigma_{\text{grad}}^2(s)] \propto \|\nabla \mathcal{E}(s_{\text{TN}})\|^2$, which is $O(1)$. This prevents the learning dynamics from vanishing. Consider the optimal Bellman operator $\mathcal{T}$
\begin{equation}
    (\mathcal{T}Q)(s, a) = R(s, a) + \gamma \mathbb{E}_{s' \sim P(\cdot | s, a)} [\max_{a'} Q(s', a')].
\end{equation}
It is well-known that $\mathcal{T}$ is a $\gamma$-contraction in the supremum norm, $\|\mathcal{T}Q_1 - \mathcal{T}Q_2\|_{\infty} \leq \gamma \|Q_1 - Q_2\|_{\infty}$.
After $k$ iterations, the error is bounded by $\|Q_k - Q^*\|_{\infty} \leq \gamma^k \|Q_0 - Q^*\|_{\infty}$. The number of iterations (and thus samples) required to reach an error $\epsilon$ is $K = \log_\gamma(\epsilon / \|Q_0 -Q^*\|_{\infty})$. For the base algorithm, $\|Q_0 - Q^*\|_{\infty}^{\text{base}} \approx |\bar{E} - E_0|/(1-\gamma) \in O(N)$. For the PGR algorithm, the generative seeding provides an initial Q-function $Q_0^{\text{PGR}}$ whose error is bounded by the energy suboptimality of the seed state $\|Q_0 - Q^*\|_{\infty}^{\text{PGR}} \approx \Delta E_{\chi} / (1-\gamma) \in O(1)$. The ratio of required iterations is,
\begin{equation}
    \frac{K_{\text{PGR}}}{K_{\text{base}}} = \frac{\log(\epsilon / \Delta E_{\chi})}{\log(\epsilon / N)} \approx \frac{\log(1/\Delta E_{\chi})}{\log N}.
\end{equation}
Each iteration requires interaction with the environment (or the generative model) to collect a batch of samples. The main bottleneck in QAS is the episode length $T$. For $\pi_{\text{base}}$, to discover a path from $s_0$ to $s^* \in \mathcal{M}_{\epsilon_c}$, the agent must perform a random walk in the discrete circuit space of depth $D$. The expected number of actions to reach the target is $T_{\text{base}} \propto |\mathcal{A}|^D = (O(N^2))^{\poly(N)}$. For $\pi_{\text{PGR}}$, the generative model $G_\psi$ directly outputs transitions $(s, a)$ that are high relevance. If we consider the synthetic replay as performing  counterfactual reasoning generating the next gate action conditioned on the target energy, it effectively reduces the depth of the required search. The agent only needs to refine the pre-existing circuit $s_{\text{TN}}$, requiring a constant depth $\Delta D$ independent of $N$. Thus $T_{\text{PGR}} \propto |\mathcal{A}|^{\Delta D} \in \poly(N)$.

Combining the iteration count $K$ and the episode length $T$, the total sample complexity is $\mathcal{S} = K \times T$. Taking the ratio,
\begin{equation}
    \frac{\mathcal{S}_{\text{PGR}}}{\mathcal{S}_{\text{base}}} = \left( \frac{\log(1/\Delta E_{\chi})}{\log N} \right) \times \left( \frac{| \mathcal{A} |^{\Delta D}}{| \mathcal{A} |^{\poly(N)}} \right) = O\left( \frac{1}{| \mathcal{A} |^{\poly(N)}} \right) = O\left( \frac{1}{N^{2\poly(N)}} \right).
\end{equation}
This is a super-exponentially decaying ratio. There exists a constant $\kappa > 1$ such that this ratio is strictly bounded above by $\kappa^{-N}$ for sufficiently large $N$, which completes the proof.

\section{GenQAS warm-start from TensorRL-QAS}
\label{sec:supp_tensorrl}

GenQAS builds on the \emph{TensorRL (fixed)} configuration of TensorRL-QAS~\cite{kundu2025tensorrlqas}. This component provides a problem-informed quantum state initialization before reinforcement learning (RL) architecture search begins. In contrast to TensorRL (trainable) the tensor network-derived warm start circuit is not included in the RL observation and its parameters are not subsequently optimized by the agent. Instead, it is retained as a fixed state-preparation prefix, while GenQAS searches for an additional circuit suffix. Thus, the complete circuit
at RL step $t$ is
\begin{equation}
    U_t = V_t U_{\mathrm{TN}},
    \label{eq:genqas_tn_composition}
\end{equation}
where $U_{\mathrm{TN}}$ is the fixed tensor network initialization and $V_t$ is the circuit constructed by the GenQAS agent up to time $t$. This design separates the physically motivated state preparation from the learned architecture refinement. In particular, the RL state encodes only $V_t$, whereas the environment evaluates the full circuit in Eq.~\eqref{eq:genqas_tn_composition}. Consequently, GenQAS retains the reduced RL-state dimensionality and reduced trainable parameter set of TensorRL (fixed), while augmenting its real experience with synthetic transitions generated by the learned dynamics model.

\subsection{MPS approximation of the target state}
\label{sec:supp_dmrg_mps}

For a target Hamiltonian $H$, we first obtain an approximate ground state using DMRG~\cite{white1992density, schollwock2011density, orus2014practical}:
\begin{equation}
    \ket{\psi_{\mathrm{MPS}}}
    \approx
    \arg\min_{\ket{\psi}}
    \frac{\bra{\psi}H\ket{\psi}}{\langle\psi|\psi\rangle}.
    \label{eq:dmrg_objective}
\end{equation}
The resulting state is represented as an $N$-site matrix product state,
\begin{equation}
    \ket{\psi_{\mathrm{MPS}}}
    =
    \sum_{i_1,\ldots,i_N}
    A_1^{i_1} A_2^{i_2} \cdots A_N^{i_N}
    \ket{i_1 i_2 \cdots i_N},
    \label{eq:mps_representation}
\end{equation}
where $i_k \in \{0,1\}$ and the tensors
$A_k^{i_k} \in \mathbb{C}^{d_{k-1}\times d_k}$ have maximum bond dimension
$\chi$. The boundary dimensions are $d_0=d_N=1$, and
$d_k \leq \chi$ for interior bonds.

The bond dimension controls the entanglement that can be represented by the
MPS and therefore determines the trade-off between classical preprocessing
cost and warm-start accuracy. DMRG optimizes the local MPS tensors through
iterative sweeps until the energy in Eq.~\eqref{eq:dmrg_objective}
converges. We use the resulting MPS only to construct the initial circuit; the MPS itself is not passed to the GenQAS agent. TensorRL-QAS employed this procedure to generate problem-aware, compact starting circuits before RL-QAS
refinement. Figure~\ref{fig:tensorRL-QAS_initializations} contrasts the trainable and fixed TensorRL-QAS initializations; GenQAS adopts the fixed configuration shown
in Fig.~\ref{fig:tensorRL-QAS_initializations}(b).

\subsection{Variational MPS-to-circuit mapping}
\label{sec:supp_mps_circuit_mapping}

We compile $\ket{\psi_{\mathrm{MPS}}}$ into a shallow parameterized circuit $U_{\mathrm{TN}}$ acting on $\ket{0}^{\otimes N}$~\cite{berezutskii2025tensor,schon2005sequential,ran2020encoding}. The circuit follows a nearest neighbour brickwork layout compatible with linearly connected quantum devices~\cite{kandala2017hardware}. It comprises $M$ two-qubit unitary blocks,
\begin{equation}
    U_{\mathrm{TN}}
    =
    U_M U_{M-1} \cdots U_1,
    \qquad U_k \in \mathrm{U}(4),
    \label{eq:brickwork_unitary}
\end{equation}
where each $U_k$ acts on an adjacent pair of qubits. We determine these blocks by maximizing the overlap between the target MPS and the prepared state, equivalently minimizing
\begin{equation}
    \mathcal{L}(\{U_k\}_{k=1}^{M})
    =
    1 -
    \left|
    \bra{\psi_{\mathrm{MPS}}}
    U_M \cdots U_1
    \ket{0}^{\otimes N}
    \right|^2.
    \label{eq:mps_overlap_loss}
\end{equation}

The brickwork layout is a one-dimensional hardware-efficient ansatz (HEA) consisting of alternating layers of nearest-neighbour two-qubit operations. Its local connectivity makes it compatible with devices arranged in a linear or effectively one-dimensional topology and reduces the need for additional routing operations. However, the practical suitability of an HEA depends on both its depth and the entanglement properties of the target-state family: deep HEAs can exhibit barren plateaus, whereas shallow HEAs can remain
trainable in favourable regimes, particularly for states obeying an area-law entanglement structure~\cite{leone2024practical}. In GenQAS, the shallow
brickwork circuit is not used as a generic variational ansatz initialized at random; rather, its two-qubit blocks are optimized to reproduce the DMRG-derived MPS. This tensor network-informed initialization provides a physically meaningful starting point for the subsequent RL-based circuit
refinement. 

\subsection{Riemannian optimization on the unitary manifold}
\label{sec:supp_riemannian_optimization}

The optimization variables in Eq.~\eqref{eq:mps_overlap_loss} belong to the
unitary manifold $\mathrm{U}(4)$:
\begin{equation}
    U_k^\dagger U_k = I_4,
    \qquad k=1,\ldots,M.
    \label{eq:unitary_constraint}
\end{equation}
We therefore use a Riemannian version of Adam with Cayley retraction, as in
Appendix~B of TensorRL-QAS. Let $G_k=\nabla_{U_k}\mathcal{L}$ denote the
Euclidean gradient obtained by automatic differentiation~\cite{baydin2018automatic}. For unitary
matrices, the Riemannian gradient can be written as
\begin{equation}
    \operatorname{grad}\mathcal{L}(U_k)
    =
    G_k - U_k G_k^\dagger U_k.
    \label{eq:riemannian_gradient}
\end{equation}
This expression lies in the tangent space of $\mathrm{U}(4)$ at $U_k$.

For each block, Riemannian Adam maintains first- and second-moment estimates~\cite{luchnikov2021riemannian},
\begin{align}
    m_{k}^{(t+1)}
    &=
    \beta_1 m_{k}^{(t)}
    + (1-\beta_1)
    \operatorname{grad}\mathcal{L}(U_k^{(t)}),
    \label{eq:riemannian_adam_momentum}\\
    v_{k}^{(t+1)}
    &=
    \beta_2 v_{k}^{(t)}
    + (1-\beta_2)
    \left\|
    \operatorname{grad}\mathcal{L}(U_k^{(t)})
    \right\|^2,
    \label{eq:riemannian_adam_variance}
\end{align}
followed by the bias-corrected tangent direction
\begin{equation}
    V_k^{(t)}
    =
    - \eta_t
    \frac{\widehat{m}_k^{(t+1)}}
    {\sqrt{\widehat{v}_k^{(t+1)}}+\varepsilon}.
    \label{eq:riemannian_adam_direction}
\end{equation}
Here, $\eta_t$ is the learning rate, $\beta_1$ and $\beta_2$ are the Adam
decay parameters, and division by the scalar denominator is applied
elementwise to the matrix-valued first moment.

To update the unitary while preserving Eq.~\eqref{eq:unitary_constraint}, we
apply Cayley retraction. Defining the skew-Hermitian matrix
\begin{equation}
    W_k^{(t)}
    =
    V_k^{(t)} {U_k^{(t)}}^\dagger
    -
    U_k^{(t)} {V_k^{(t)}}^\dagger,
    \label{eq:cayley_generator}
\end{equation}
the updated block is
\begin{equation}
    U_k^{(t+1)}
    =
    \left(I_4-\frac{1}{2}W_k^{(t)}\right)^{-1}
    \left(I_4+\frac{1}{2}W_k^{(t)}\right)
    U_k^{(t)}.
    \label{eq:cayley_retraction}
\end{equation}
Because $W_k^{(t)}$ is skew-Hermitian, the Cayley map in
Eq.~\eqref{eq:cayley_retraction} preserves unitarity up to numerical precision~\cite{luchnikov2021riemannian}. The momentum variables are subsequently transported to the tangent
space at the updated point before the next iteration. All two-qubit blocks are updated simultaneously until the overlap loss converges. The resulting MPS-derived brickwork circuit provides the warm-start circuit used by both TensorRL-QAS configurations in Fig.~\ref{fig:tensorRL-QAS_initializations}; GenQAS retains only the fixed
variant in Fig.~\ref{fig:tensorRL-QAS_initializations}(b). 
\begin{figure}[h!]
    \centering
    \includegraphics[width=\linewidth]{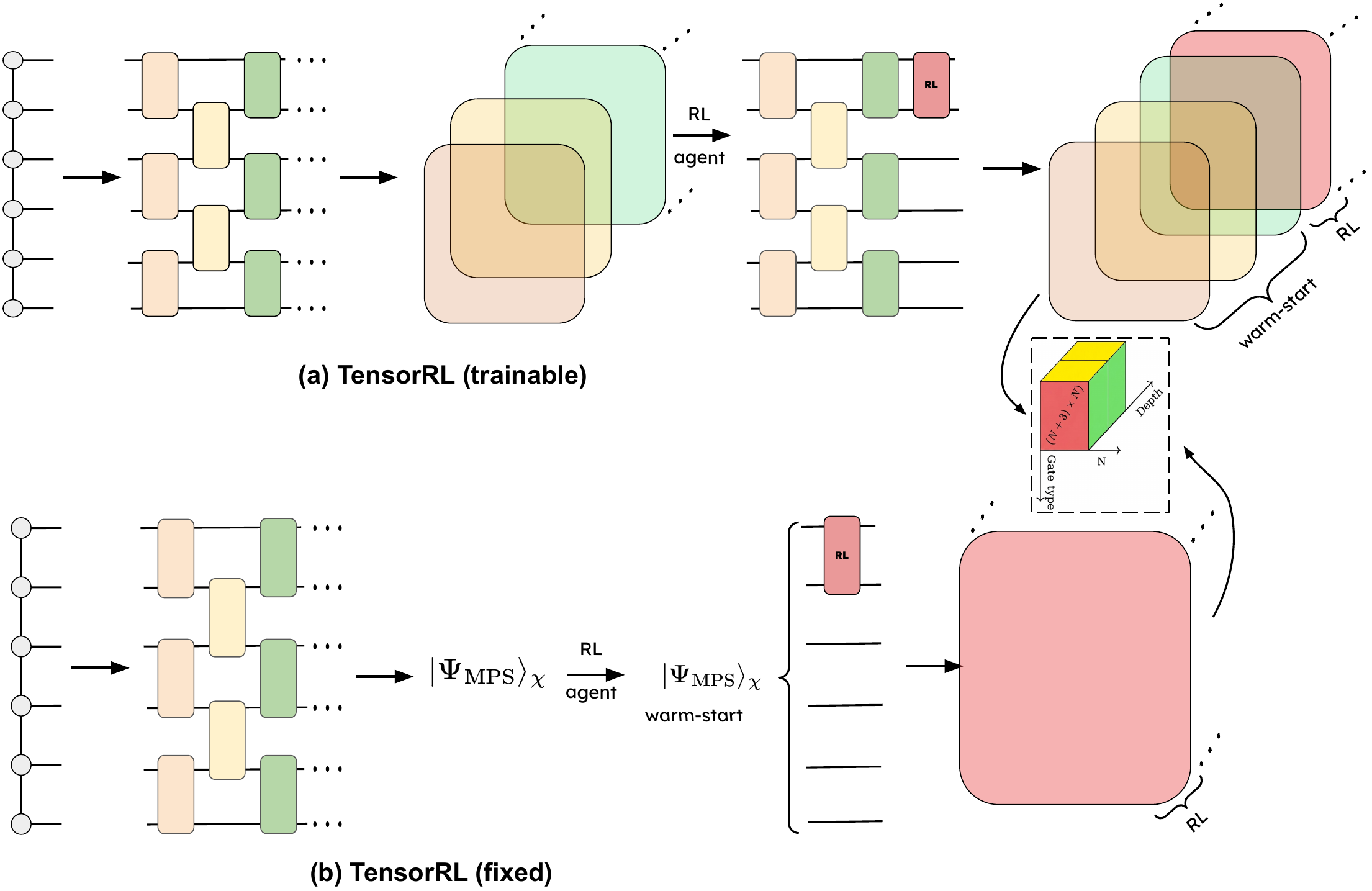}
    \caption{\textbf{Tensor network initializations in TensorRL-QAS.}
(a) \emph{TensorRL (trainable)}: the DMRG-derived MPS is mapped to a brickwork circuit and encoded in the RL state. The RL agent can subsequently modify the warm start circuit and optimize its parameters jointly with those of the newly appended gates. (b) \emph{TensorRL-fixed TN-init}: the
MPS-derived circuit prepares a fixed warm-start state
$\ket{\psi_{\mathrm{MPS}}}_{\chi}$, while the RL state encodes only the circuit appended by the agent. The fixed initialization is excluded from the agent observation and remains non-trainable throughout RL search. GenQAS uses
the TensorRL (fixed) configuration in (b), upon which the agent constructs the learnable circuit suffix $V_t$. Adapted from TensorRL-QAS~\cite{kundu2025tensorrlqas}.}
    \label{fig:tensorRL-QAS_initializations}
\end{figure}

\subsection{Fixed initialization within GenQAS}
\label{sec:supp_fixed_tensorrl_genqas}

After convergence of the mapping objective, the two-qubit unitaries are decomposed into the native search and execution gate set. As illustrated in Fig.~\ref{fig:tensorRL-QAS_initializations}(b), $U_{\mathrm{TN}}$ prepares the warm-start state but is excluded from the RL-state and is not modified by the GenQAS agent. In the TensorRL-QAS configuration adopted here, the compiled circuit is expressed in terms of single-qubit rotations and entangling gates, namely $\{\texttt{RX},\texttt{RY},\texttt{RZ},\texttt{CNOT}\}$. The resulting circuit is denoted $U_{\mathrm{TN}}$ and remains unchanged throughout a GenQAS training run.

Specifically, the initial quantum state used by the environment is
\begin{equation}
    \ket{\psi_0}
    =
    U_{\mathrm{TN}}\ket{0}^{\otimes N}.
    \label{eq:fixed_initial_state}
\end{equation}
At time $t$, an action selected by the agent appends a gate to the learnable suffix $V_t$. The energy, reward, and termination condition are evaluated using
\begin{equation}
    E_t
    =
    \bra{0}
    U_{\mathrm{TN}}^\dagger V_t^\dagger
    H
    V_t U_{\mathrm{TN}}
    \ket{0}.
    \label{eq:full_energy_fixed_tn}
\end{equation}

Crucially, the circuit representation presented to the agent contains only the appended architecture $V_t$. If the RL suffix has maximum depth $D$, the state dimension is therefore independent of the warm start depth $D_{\mathrm{TN}}$. In contrast, a trainable TN initialization would encode both the TN prefix and the subsequently generated gates, increasing the representation from the suffix-only form to one that includes $D_{\mathrm{TN}}+D$ circuit layers. Retaining $U_{\mathrm{TN}}$ as a fixed
prefix reduces the observation size, avoids optimizing its parameters during each RL episode, and maintains a lower-cost environment for the generative replay agent. 

\subsection{Relation to the GenQAS contribution}
\label{sec:supp_relation_tensorrl_genqas}

The tensor network warm start is inherited from TensorRL-QAS (fixed) and serves only as a fixed, problem informed initialization. GenQAS differs from TensorRL-QAS in the experience-learning mechanism after initialization: instead of training exclusively on transitions sampled from the real replay memory, GenQAS learns a transition-reward-termination model from real experience and uses it to generate synthetic transitions on demand.

Accordingly, the tensor network module affects the initial state distribution and the circuit evaluation prefix, whereas the generative replay module affects the distribution of transitions used for DDQN updates. The two components are complementary: $U_{\mathrm{TN}}$ reduces the burden of
discovering a useful initial circuit from scratch, while model-based generative replay increases the effective density of training data in the remaining architecture search problem.

\subsection{Warm-start resource cost and GenQAS refinement}\label{sec:supp_warmstart_resources}

Table~\ref{tab:tn_warmstart_resources} reports the resource requirements of
the fixed MPS-derived warm-start circuits after decomposition into the native
gate set $\{\texttt{RX},\texttt{RY},\texttt{RZ},\texttt{CNOT}\}$. Here,
$\mathrm{ROT}=N_{\texttt{RX}}+N_{\texttt{RY}}+N_{\texttt{RZ}}$ denotes the
total number of single-qubit rotations. The compiled warm-start circuits have
a fixed depth of 27 across the considered systems, while the number of
two-qubit gates increases from 15 CNOTs for 6-qubit $\mathrm{BeH}_2$ to 42
CNOTs for the 15-qubit transverse field Ising model (TFIM).
\begin{table*}[h!]
\caption{\textbf{Resource cost and energy error of the fixed MPS-derived
warm-start circuits.} The MPS obtained by DMRG is mapped to a brickwork
circuit and decomposed into the native gate set
$\{\texttt{RX},\texttt{RY},\texttt{RZ},\texttt{CNOT}\}$. ROT denotes the
total number of single-qubit rotations. The initialization error is
$\epsilon_{\mathrm{TN}}=|E_{\mathrm{TN}}-E_0|$, and
$\epsilon_{\mathrm{GenQAS}}$ is the lowest error obtained by GenQAS.}
\label{tab:tn_warmstart_resources}
\centering
\begin{tabular}{lllllll}
\hline
System & ROT & CNOT & Depth &
$\epsilon_{\mathrm{TN}}$ &
$\epsilon_{\mathrm{GenQAS}}$ &
$R_{\epsilon}$ \\
\hline
6-$\mathrm{BeH}_2$
& 93 & 15 & 27 & $5.96\times10^{-3}$
& $2.3\times10^{-5}$ & $99.6\%$ \\

8-$\mathrm{H}_2\mathrm{O}$
& 129 & 21 & 27 & $2.70\times10^{-3}$
& $1.2\times10^{-3}$ & $55.6\%$ \\

10-$\mathrm{H}_2\mathrm{O}$
& 165 & 27 & 27 & $4.77\times10^{-3}$
& $4.1\times10^{-4}$ & $91.4\%$ \\

12-$\mathrm{H}_2\mathrm{O}$
& 196 & 32 & 27 & $5.07\times10^{-1}$
& $2.2\times10^{-2}$ & $95.7\%$ \\

15-TFIM
& 255 & 42 & 27 & $1.025$
& $1.5\times10^{-2}$ & $98.5\%$ \\
\hline
\end{tabular}
\end{table*}

To quantify refinement beyond the fixed initialization, we define the
warm-start and final GenQAS errors relative to the reference ground-state
energy $E_0$ as
\begin{equation}
    \epsilon_{\mathrm{TN}}
    =
    \left|E_{\mathrm{TN}}-E_0\right|,
    \qquad
    \epsilon_{\mathrm{GenQAS}}
    =
    \left|E_{\mathrm{GenQAS}}-E_0\right|.
    \label{eq:warmstart_energy_errors}
\end{equation}
The absolute error reduction achieved by GenQAS is
\begin{equation}
    \Delta \epsilon
    =
    \epsilon_{\mathrm{TN}}
    -
    \epsilon_{\mathrm{GenQAS}},
    \label{eq:warmstart_absolute_reduction}
\end{equation}
and the relative error reduction is
\begin{equation}
    R_{\epsilon}
    =
    100
    \times
    \frac{\epsilon_{\mathrm{TN}}-\epsilon_{\mathrm{GenQAS}}}
    {\epsilon_{\mathrm{TN}}}.
    \label{eq:warmstart_relative_reduction}
\end{equation}

GenQAS improves upon the fixed tensor network initialization across all
evaluated systems. For example, for 12-qubit $\mathrm{H}_2\mathrm{O}$, GenQAS
reduces the MPS-to-circuit mapping error from $5.07\times10^{-1}$ to
$2.2\times10^{-2}$. For the 15-qubit TFIM, GenQAS reduces the initialization
error from $1.025$ to $1.5\times10^{-2}$, corresponding to an absolute reduction of $1.010$ and a relative reduction of $98.5\%$. These results demonstrate that the MPS-derived circuit provides a useful but not necessarily sufficient warm start, which is subsequently refined through RL-based architecture search.

\section{Clifford synthesis}
\label{sec:supp_clifford_synthesis_exp}

We evaluate GenQAS on a discrete 4-qubit Clifford synthesis task in which the
agent constructs a circuit that exactly implements a target Clifford operator.
This experiment complements the molecular and TFIM benchmarks by removing
variational parameter optimization and energy-expectation evaluation. The
target is instead specified by an exact Clifford transformation, and success
is determined through exact equality of the synthesized and target Clifford
operators.

Our formulation is conceptually related to prior reinforcement learning
approaches to Clifford circuit synthesis, which also frame synthesis as a sequential decision problem over a Clifford gate set and use stabilizer tableau representations of the target operator~\cite{kremer2024practical}. However,
our environment is independently implemented and differs in several important respects. In particular, we construct circuits forward from the identity, represent both the current and target Clifford operators in the observation, use QASM-defined target circuits, include the tableau phase information, and train the GenQAS DDQN agent with prioritized generative replay.

\subsection{Target circuits and action space}

Target Clifford circuits are loaded from QASM files organized by qubit number and circuit depth category and generated using the QCircuitBench dataset~\cite{yang2026qcircuitbench}. At the beginning of each episode, one target circuit is sampled uniformly from the dataset entries. The target circuit is converted to a Qiskit \texttt{Clifford} object, denoted $C_{\mathrm{tar}}$, while the agent starts from the identity circuit with Clifford representation $C_0 = I$. The agent sequentially appends gates to the current circuit until the synthesized Clifford operator exactly equals $C_{\mathrm{tar}}$ or the maximum synthesis horizon $T_{\max}$ is reached. We use
$T_{\max}=20$ in the reported Clifford synthesis experiments. The action space contains
\begin{equation}
    \mathcal{A}_{\mathrm{Cliff}}
    =
    \left\{
        \texttt{H}_q,\,
        \texttt{S}_q,\,
        \texttt{S}_q^{\dagger},\,
        \texttt{CNOT}_{c,t}
    \right\},
    \qquad
    q\in\{0,\ldots,N-1\},
    \quad
    c\neq t.
    \label{eq:clifford_action_space}
\end{equation}
The resulting action-space size is
\begin{equation}
    |\mathcal{A}_{\mathrm{Cliff}}|
    =
    3N + N(N-1).
    \label{eq:clifford_action_size}
\end{equation}
For the 4-qubit benchmark, this gives $24$ possible actions: $12$
single-qubit Clifford actions and $12$ directed \texttt{CNOT} actions.

We apply local action masking to suppress immediately repeated or cancelling gate sequences on the same qubit. In particular, the environment masks consecutive identical $H$, $S$, or $S^{\dagger}$ actions on a qubit, masks consecutive $S$ and $S^{\dagger}$ actions on the same qubit, and masks a repeated CNOT with the same ordered control-target pair. All actions are masked once the maximum episode length is reached. These constraints reduce locally redundant trajectories while retaining the full Clifford synthesis objective.

\subsection{State representation}

A Clifford operator on $N$ qubits is represented using the stabilizer tableau and a phase vector obtained by the Qiskit \texttt{Clifford} class. Let $\mathsf{T}(C)$ denote the binary tableau representation of a Clifford operator $C$, and let $\mathbf{p}(C)$ denote the associated phase vector. We
define the vectorized Clifford representation as
\begin{equation}
    \mathbf{v}(C)
    =
    \left[
        \operatorname{vec}(\mathsf{T}(C)),
        \operatorname{vec}(\mathbf{p}(C))
    \right].
    \label{eq:clifford_vector_representation}
\end{equation}
Unlike representations that omit phase information, the present environment retains both tableau and phase entries because exact operator equality, including phase, defines successful synthesis. At time step $t$, the observation supplied to the agent consists of the
target Clifford representation, the current Clifford representation, and a normalized progress variable:
\begin{equation}
    s_t
    =
    \left[
        \mathbf{v}(C_{\mathrm{tar}}),
        \mathbf{v}(C_t),
        \frac{t}{T_{\max}}
    \right].
    \label{eq:clifford_state}
\end{equation}
This construction makes the target transformation, the operator synthesized
so far, and the remaining effective episode horizon available to the DDQN
agent.

\subsection{Reward and termination}

To provide dense feedback before exact synthesis is achieved, we compare the
current and target Clifford tableaux entrywise. Let
$\mathsf{T}_t$, $\mathbf{p}_t$ and
$\mathsf{T}_{\mathrm{tar}}$, $\mathbf{p}_{\mathrm{tar}}$ denote the tableau
and phase vector of the current and target Clifford operators, respectively.
The tableau-match fraction is defined as
\begin{equation}
    m_t
    =
    \frac{
        \sum_{i}
        \mathbb{I}
        \left[
            \mathsf{T}_{t,i}
            =
            \mathsf{T}_{\mathrm{tar},i}
        \right]
        +
        \sum_{j}
        \mathbb{I}
        \left[
            p_{t,j}
            =
            p_{\mathrm{tar},j}
        \right]
    }{
        |\mathsf{T}_t| + |\mathbf{p}_t|
    }.
    \label{eq:clifford_match_fraction}
\end{equation}
The associated tableau-distance proxy is
\begin{equation}
    d_t = 1-m_t.
    \label{eq:clifford_tableau_distance}
\end{equation}

The reward after applying action $a_t$ is
\begin{equation}
    r_t
    =
    -\lambda_{\mathrm{step}}
    -
    \lambda_{\mathrm{CX}}
    \mathbb{I}[a_t\in\mathrm{CNOT}]
    +
    \lambda_{\mathrm{partial}}
    \left(
        m_{t+1}-m_t
    \right)
    +
    R_{\mathrm{solve}}
    \mathbb{I}[C_{t+1}=C_{\mathrm{tar}}],
    \label{eq:clifford_reward}
\end{equation}
where
\begin{equation}
    \lambda_{\mathrm{step}}=0.01,
    \qquad
    \lambda_{\mathrm{CX}}=0.02,
    \qquad
    \lambda_{\mathrm{partial}}=0.5,
    \qquad
    R_{\mathrm{solve}}=10.
    \label{eq:clifford_reward_hyperparameters}
\end{equation}
The first term penalizes long circuits, the second applies an additional cost
to entangling operations, the third rewards local improvement in tableau
agreement, and the final term provides a sparse terminal reward for exact
Clifford synthesis.

An episode terminates when
\begin{equation}
    C_t=C_{\mathrm{tar}}
\end{equation}
or when $t=T_{\max}$. Exact equality is evaluated using the Qiskit Clifford
representation, rather than through an approximate state or unitary-distance
threshold. Therefore, a successful episode corresponds to exact synthesis of
the target Clifford operation within the allowed number of construction
steps.

\section{GenQAS implementation details}
\label{sec:genqas_implement_details}

\subsection{Hyperparameters}

GenQAS was implemented as a prioritized generative replay extension of a Double Deep Q-Network (DDQN) agent as discussed in the main text Sec~\ref{sec:methods}. The agent used a five-layer multilayer perceptron with $1000$ neurons per hidden layer, LeakyReLU activations, no dropout, and an ADAM learning rate of $3\times10^{-4}$. All experiments employed a replay memory of $20000$ transitions and mini-batches of $1000$ transitions. The target network was updated every $100$ optimization steps. The replay buffer was prioritized, and the conditional generative transition model was trained every $10$ steps using real replay transitions. Synthetic experience was generated on-the-fly and mixed with real transitions at generation ratios $G_r\in\{0.2,0.4,0.6,0.8\}$; all other agent and environment settings were kept fixed within each system size. For the 10- and 12-qubit H$_2$O tasks, we used an $n$-step DDQN update with $n=5$; the same $n=5$ roll-out setting was used throughout the reported GenQAS experiments for consistency.

The exploration policy was initialized with $\epsilon=1.0$ and decayed multiplicatively with factor $0.99995$ until reaching $\epsilon_{\min}=0.05$. We used a per-layer discount factor derived from a final discount value of $\gamma_{\mathrm{final}}=0.005$. Circuit construction was initialized from a tensor-network prior obtained using bond dimension $\chi=2$, with no zero-parameter initialization. All noiseless results were obtained using exact expectation-value evaluations ($n_{\mathrm{shots}}=0$), without error mitigation or stochastic early termination. Molecular Hamiltonians were constructed using the Jordan-Wigner mapping with tapering enabled, and continuous circuit parameters were optimized after each architectural modification using COBYLA.

\subsection{Computational resources}
\label{sec:computational_resources}

All experiments were executed on the DelftBlue high-performance computing system~\cite{DHPC2024}. The GenQAS training and simulation workloads used the \texttt{gpu-a100} partition, comprising nodes with two Intel Xeon Gold 6448Y processors ($32$ cores per processor), four NVIDIA Tesla A100 GPUs with $80$ GB memory per GPU, $500$ GB system memory, and $1.5$ TB local storage. Unless otherwise
stated, each training run used $2$ NVIDIA A100 GPU(s) and
$4$ CPU core(s). For some of the comparative analysis for GenQAS, we utilise a single node system equipped with NVIDIA RTX A6000. The molecular Hamiltonian construction,
tensor network initialization, RL training, and Qulacs statevector simulations were executed using the computational resources allocated to each run.

The source code, complete hyperparameter configurations, random seeds, and instructions for reproducing the experiments will be made available in the associated code repository upon publication.

\begin{table}[t]
\centering
\caption{Main GenQAS hyperparameters for the noiseless molecular benchmarks. The generation ratio $G_r$ was swept over $\{0.2,0.4,0.6,0.8\}$ while all remaining settings were held fixed for a given qubit number.}
\label{tab:genqas-hyperparameters}
\begin{tabular}{lcccc}
\hline
\textbf{Hyperparameter} & \textbf{6-qubit} & \textbf{8-qubit} & \textbf{10-qubit} & \textbf{12-qubit} \\
\hline
System & BeH$_2$ & H$_2$O & H$_2$O & H$_2$O \\
Training episodes & $4000$ & $7000$ & $5000$ & $5000$ \\
Maximum circuit layers & $47$ & $47$ & $107$ & $71$ \\
Acceptance error & $1.6\times10^{-3}$ & $1.6\times10^{-3}$ & $1.6\times10^{-3}$ & $2.54\times10^{-2}$ \\
Generation ratio, $G_r$ & $\{0.2,0.4,0.6,0.8\}$ & $\{0.2,0.4,0.6,0.8\}$ & $\{0.2,0.4,0.6,0.8\}$ & $\{0.2,0.4,0.6,0.8\}$ \\
Replay memory size & $20000$ & $20000$ & $20000$ & $20000$ \\
Batch size & $1000$ & $1000$ & $1000$ & $1000$ \\
Hidden layers & $[1000]^5$ & $[1000]^5$ & $[1000]^5$ & $[1000]^5$ \\
Learning rate & $3\times10^{-4}$ & $3\times10^{-4}$ & $3\times10^{-4}$ & $3\times10^{-4}$ \\
Generator training frequency & $10$ & $10$ & $10$ & $10$ \\
Prioritized replay & Yes & Yes & Yes & Yes \\
$n$-step return & $5$ & $5$ & $5$ & $5$ \\
Target-network update & $100$ & $100$ & $100$ & $100$ \\
$\epsilon$ decay / minimum & $0.99995 / 0.05$ & $0.99995 / 0.05$ & $0.99995 / 0.05$ & $0.99995 / 0.05$ \\
Final discount, $\gamma_{\mathrm{final}}$ & $0.005$ & $0.005$ & $0.005$ & $0.005$ \\
TN initialization / bond dimension & Yes / $2$ & Yes / $2$ & Yes / $2$ & Yes / $2$ \\
Optimizer & COBYLA & COBYLA & COBYLA & COBYLA \\
\hline
\end{tabular}
\end{table}

\subsection{Molecular Hamiltonian construction}
\label{sec:molecular_hamiltonians}

We generated the molecular electronic-structure Hamiltonians using the quantum chemistry module of PennyLane (\texttt{qml.qchem})~\cite{bergholm2018pennylane,arrazola2021differentiable}. For each molecule, we specify the atomic symbols, Cartesian nuclear coordinates, total charge $Q=0$, spin multiplicity $\mathrm{mult}=1$, and an active-space definition characterized by the number of active electrons and active orbitals. The resulting second-quantized electronic Hamiltonian was mapped to a qubit Hamiltonian using the Jordan-Wigner transformation,
\begin{equation}
H = \sum_{j} c_j P_j,
\end{equation}
where $c_j\in\mathbb{R}$ are Pauli coefficients and $P_j\in\{I,X,Y,Z\}^{\otimes N}$ are $N$-qubit Pauli words. The Pauli decomposition $\{(c_j,P_j)\}_j$ was stored and used as the target Hamiltonian in the variational energy objective,
\begin{equation}
E(\boldsymbol{\theta},\mathcal{A})
=
\langle 0|
U^{\dagger}(\boldsymbol{\theta},\mathcal{A})
H
U(\boldsymbol{\theta},\mathcal{A})
|0\rangle,
\end{equation}
where $\mathcal{A}$ denotes the circuit architecture generated by the RL agent.

We considered $\behtwo$, and $\htwoo$ molecules. For $\behtwo$, the nuclear geometry was specified as H$(0,0,-1.330)$-Be$(0,0,0)$-H$(0,0,1.330)$, with four active electrons in three active orbitals. For $\htwoo$, we used the geometry H$(-0.021,-0.002,0.000)$, O$(0.835,0.452,0.000)$, and H$(1.477,-0.273,0.000)$, with four active electrons. The standard H$_2$O instances used four active orbitals; larger active-space instances used the \texttt{6-31g} basis to support the corresponding higher 10-qubit  Hamiltonians. Unless a larger H$_2$O active space was required, molecular integrals were calculated in the \texttt{STO-3G} basis.

Our implementation used PennyLane for molecular Hamiltonian generation, PennyLane NumPy for differentiable array handling, and SciPy linear algebra routines for exact diagonalization of the Hamiltonian matrix. We additionally constructed an equivalent \texttt{SparsePauliOp} representation using Qiskit Quantum Information~\cite{javadi2024quantum}. This enabled an independent consistency check between the PennyLane and Qiskit representations: we verified the agreement of the dense Hamiltonian matrices and their ground-state eigenvalues before storing each Hamiltonian. The final data file contained the dense Hamiltonian matrix, its eigenvalue spectrum, the Pauli strings, the associated coefficients, and a zero energy shift. This validation ensures that the Pauli Hamiltonian used by the RL-QAS environment is consistent with the molecular Hamiltonian generated by the quantum chemistry workflow.

\subsection{Quantum simulators}
\label{sec:quantum_simulation_details}

Although the molecular Hamiltonians were generated using PennyLane, the tensor-network warm-start component inherited from TensorRL-QAS was implemented using Qiskit~\cite{javadi2024quantum} and Quimb~\cite{gray2018quimb}. In particular, Quimb was used to construct and optimize the matrix product state (MPS) approximation of the target ground state with bond dimension $\chi=2$. The resulting MPS was subsequently converted into a quantum circuit representation using the TensorRL-QAS warm-start procedure, implemented through Qiskit. This circuit provides the physics-informed initial architecture for the RL agent before sequential gate level architecture search begins.

All quantum state simulations and energ expectation value evaluations reported in this work were performed using the Qulacs~\cite{suzuki2021qulacs}. For a circuit architecture $\mathcal{A}$ and variational parameters $\boldsymbol{\theta}$, Qulacs prepares the state
\begin{equation}
\ket{\psi(\boldsymbol{\theta},\mathcal{A})}
=
U(\boldsymbol{\theta},\mathcal{A})\ket{0}^{\otimes N},
\end{equation}
and evaluates the molecular energy as
\begin{equation}
E(\boldsymbol{\theta},\mathcal{A})
=
\bra{\psi(\boldsymbol{\theta},\mathcal{A})}
H
\ket{\psi(\boldsymbol{\theta},\mathcal{A})}
=
\sum_j c_j
\bra{\psi(\boldsymbol{\theta},\mathcal{A})}
P_j
\ket{\psi(\boldsymbol{\theta},\mathcal{A})}.
\end{equation}
Here, $H=\sum_j c_jP_j$ is the Jordan-Wigner-mapped Pauli Hamiltonian. Qulacs was used consistently throughout the RL-training loop.

\subsection{Mechanistic signal-density analysis}
\label{sec:supp_signal_density}

This section provides a mechanistic interpretation of how the generation
ratio controls the effective density of useful transitions during GenQAS
training. The analysis is conditional on the assumptions of
Theorems~\ref{thm:sparsity} and~\ref{thm:haar},
Lemma~\ref{lem:fidelity}, and Theorem~\ref{thm:convergence}. It does not
assert that every generated transition is useful; rather, it identifies the
conditions under which synthetic replay can compensate for the decreasing
density of useful real transitions in large QAS problems.

Let $\mathcal{G}_{\epsilon}$ denote the set of circuit states with energy
within $\epsilon$ of the ground state energy $E_0$. For a replay distribution
$\mathcal{D}$, we define the useful-transition density as
\begin{equation}
    \rho(\mathcal{D})
    =
    \Pr_{\tau\sim\mathcal{D}}
    \left[
        s(\tau)\in\mathcal{G}_{\epsilon}
    \right],
    \label{eq:supp_signal_density_definition}
\end{equation}
where $s(\tau)$ denotes the circuit state associated with transition
$\tau$. This quantity measures the probability that a transition sampled for
a DDQN update originates from the target low-energy region.

For passive replay, the replay buffer contains only transitions collected
through interaction with the QAS environment. Let
$\rho_{\mathrm{real}}(N)$ denote the useful-transition density in this
distribution for an $N$ qubit problem. Under the weakly informed exploration
and Hamiltonian assumptions used in Theorems~\ref{thm:sparsity} and
\ref{thm:haar}, this density satisfies
\begin{equation}
    \rho_{\mathrm{real}}(N)
    \leq
    \exp\left[
        -c\,2^N\log\left(\frac{1}{\epsilon}\right)
    \right],
    \qquad c>0.
    \label{eq:supp_passive_signal_density}
\end{equation}
Thus, as the number of qubits increases, useful transitions constitute a
rapidly decreasing fraction of the real interaction history. Uniform replay
and prioritized experience replay can alter the sampling frequency of
previously observed transitions, but neither can change the probability that a
useful transition was collected initially.

At each DDQN update, GenQAS constructs a batch of total size $B$ by mixing
real transitions sampled from the prioritized replay memory with synthetic
transitions generated on demand. The generation ratio $G_r\in[0,1)$ denotes
the synthetic fraction of the total update batch:
\begin{equation}
    B_{\mathrm{real}}
    =
    \left\lfloor
        (1-G_r)B
    \right\rfloor,
    \qquad
    B_{\mathrm{syn}}
    =
    B-B_{\mathrm{real}}.
    \label{eq:supp_generation_ratio}
\end{equation}
Consequently, $G_r=0.4$ produces an update batch containing approximately $60\%$ real and $40\%$ synthetic transitions, while $G_r=0.8$ produces approximately $20\%$ real and $80\%$ synthetic transitions. The synthetic samples are generated from state action seeds drawn from real replay memory and are not stored persistently. The replay buffer therefore remains a record of real environment interactions.

Let $\rho_{\mathrm{syn}}(N)$ denote the useful transition density among synthetic samples. Lemma~\ref{lem:fidelity} shows that the MPS derived warm start can be prepared within trace distance $\delta$ of the relevant low energy manifold under the stated MPS approximation and circuit mapping conditions. Let $q\in[0,1]$ denote the probability that the learned local transition model preserves this useful neighborhood when generating a synthetic transition. Under this model alignment condition,
\begin{equation}
    \rho_{\mathrm{syn}}(N)
    \geq
    q(1-\delta).
    \label{eq:supp_synthetic_signal_density}
\end{equation}
The factor $q$ separates the fidelity of the MPS derived initialization from the local predictive accuracy of the learned transition model. Accordingly, Eq.~\eqref{eq:supp_synthetic_signal_density} is a theoretical condition, not a guarantee that follows automatically from neural network training. The
generator and critic diagnostics in Figures~\ref{fig:gen-losses} and \ref{fig:q-loss} assess whether learning remains stable for the investigated
generation ratios.

The useful-transition density of the mixed GenQAS update distribution is
\begin{equation}
    \rho_{\mathrm{mix}}(N)
    =
    (1-G_r)\rho_{\mathrm{real}}(N)
    +
    G_r\rho_{\mathrm{syn}}(N).
    \label{eq:supp_mixed_signal_density}
\end{equation}
Combining Eqs.~\eqref{eq:supp_passive_signal_density} and
\eqref{eq:supp_synthetic_signal_density} gives
\begin{equation}
    \rho_{\mathrm{mix}}(N)
    \geq
    (1-G_r)\rho_{\mathrm{real}}(N)
    +
    G_rq(1-\delta).
    \label{eq:supp_mixed_signal_lower_bound}
\end{equation}
The corresponding gain relative to passive replay is
\begin{align}
    \frac{
        \rho_{\mathrm{mix}}(N)
    }{
        \rho_{\mathrm{real}}(N)
    }
    &\geq
    (1-G_r)
    +
    G_rq(1-\delta)
    \frac{
        1
    }{
        \rho_{\mathrm{real}}(N)
    }
    \nonumber\\
    &\geq
    (1-G_r)
    +
    G_rq(1-\delta)
    \exp\left[
        c\,2^N\log\left(\frac{1}{\epsilon}\right)
    \right].
    \label{eq:supp_signal_density_gain}
\end{align}
Equation~\eqref{eq:supp_signal_density_gain} describes the proposed
mechanism of generative replay. If the MPS warm start remains sufficiently
accurate, such that $\delta$ is small, and the learned dynamics model
preserves local useful-transition structure with nonzero probability $q$, the
synthetic contribution does not vanish with system size. This contrasts with
the useful density of passive replay, which decreases rapidly under the
assumptions of Eq.~\eqref{eq:supp_passive_signal_density}.

Finally, consider a regime in which a policy requires $K$ useful replay
transitions to achieve a fixed reduction in value-estimation error. The
expected number of sampled replay transitions required to obtain these
transitions scales approximately as
\begin{equation}
    \mathbb{E}[M]
    \approx
    \frac{K}{\rho}.
    \label{eq:supp_useful_sample_requirement}
\end{equation}
Therefore,
\begin{equation}
    \frac{
        \mathbb{E}[M_{\mathrm{real}}]
    }{
        \mathbb{E}[M_{\mathrm{mix}}]
    }
    \approx
    \frac{
        \rho_{\mathrm{mix}}(N)
    }{
        \rho_{\mathrm{real}}(N)
    }.
    \label{eq:supp_replay_complexity_ratio}
\end{equation}
Together, Eqs.~\eqref{eq:supp_signal_density_gain} and
\eqref{eq:supp_replay_complexity_ratio} provide an intuitive connection
between useful-transition densification and the episode complexity result in
Theorem~\ref{thm:convergence}. The finite-size results in
Figure~\ref{fig:scaling} do not establish an asymptotic scaling law, but they
are consistent with the predicted increase in the relative value of
generative replay as the passive replay signal becomes sparse.

\end{appendices}

\end{document}